\documentclass[11pt]{article}

\usepackage{graphicx}
\usepackage{times}
\usepackage{amsmath}               
\usepackage{amsfonts}              
\usepackage{amsthm}                
\usepackage{amssymb}
\usepackage{todonotes}
\usepackage[pagebackref]{hyperref}
\usepackage[margin=1in]{geometry}

\newtheorem{thm}{Theorem}[section]
\newtheorem{lem}[thm]{Lemma}
\newtheorem{claim}[thm]{Claim}

\newtheorem{cor}[thm]{Corollary}
\newtheorem{conj}[thm]{Conjecture}
\newtheorem{defi}[thm]{Definition}
\newtheorem{rem}[thm]{Remark}

\renewcommand{\leq}{\leqslant}

\renewcommand{\geq}{\geqslant}

\DeclareMathOperator*{\E}{\mathbb{E}}
\DeclareMathOperator*{\secondmax}{max2}
\DeclareMathOperator{\Inf}{\mathsf{Inf}}
\DeclareMathOperator*{\PP}{\mathbf{P}}

\newcommand{\calD}{\mathcal{D}}
\newcommand{\calA}{\mathcal{A}}
\newcommand{\calV}{\mathcal{V}}

\newcommand{\wgmd}{w_{\mathsf{GMD}}}
\newcommand{\wgp}{w_{\mathsf{GP}}}
\newcommand{\ndeg}{\mathsf{ndeg}}
\newcommand{\Valug}{\mathsf{Val_{UG}}}
\newcommand{\Optug}{\mathsf{Opt_{UG}}}
\newcommand{\Valgmd}{\mathsf{Val_{GMD}}}
\newcommand{\Optgmd}{\mathsf{Opt_{GMD}}}
\newcommand{\Valgp}{\mathsf{Val_{GP}}}
\newcommand{\Optgp}{\mathsf{Opt_{GP}}}
\newcommand{\Val}{\mathsf{Val}}
\newcommand{\Opt}{\mathsf{Opt}}
\newcommand{\RR}{\mathbb{R}}      
\newcommand{\ZZ}{\mathbb{Z}}      

\begin{document}
\title{{\bf Hardness of Graph Pricing \\ Through Generalized Max-Dicut}}

\author{
Euiwoong Lee\thanks{Supported by a Samsung Fellowship, Venkatesan Guruswami's US-Israel BSF grant 2008293, and NSF CCF-1115525. Most of this work is done while visiting Microsoft Research New England. {\tt euiwoonl@cs.cmu.edu}}}

\date{Computer Science Department \\ Carnegie Mellon University \\ Pittsburgh, PA 15213.}

\maketitle
\thispagestyle{empty}

\begin{abstract}
The Graph Pricing problem is among the fundamental problems whose approximability is not well-understood.
While there is a simple combinatorial $\frac{1}{4}$-approximation algorithm, 
the best hardness result remains at $\frac{1}{2}$ assuming the Unique Games Conjecture (UGC). 
We show that it is NP-hard to approximate within a factor better than $\frac{1}{4}$ under the UGC, so that the simple combinatorial algorithm might be the best possible. 
We also prove that for any $\epsilon > 0$, there exists $\delta > 0$ such that the integrality gap of $n^{\delta}$-rounds of the Sherali-Adams hierarchy of linear programming for Graph Pricing is at most $\frac{1}{4} + \epsilon$.

\smallskip
This work is based on the effort to view the Graph Pricing problem as a Constraint Satisfaction Problem (CSP) simpler than the standard and complicated formulation. 
We propose the problem called Generalized Max-Dicut($T$), which has a domain size $T + 1$ for every $T \geq 1$. Generalized Max-Dicut(1) is well-known Max-Dicut. 
There is an approximation-preserving reduction from Generalized Max-Dicut on directed acyclic graphs (DAGs) to Graph Pricing,
and both our results are achieved through this reduction. 
Besides its connection to Graph Pricing, the hardness of Generalized Max-Dicut is interesting in its own right since in most arity two CSPs studied in the literature, SDP-based algorithms perform better than LP-based or combinatorial algorithms --- for this arity two CSP, a simple combinatorial algorithm does the best. 

\end{abstract}

\newpage

\section {Introduction}

Consider the following natural problem for a seller with a profit-maximization objective. The seller has $n$ types of items $1, \dots, n$, each with unlimited copies, and there are $m$ customers $1, \dots, m$. Each customer $j$ has her own budget $b_j$ and a subset of items $e_j \subseteq \left\{ 1, \dots, n \right\}$ that she is interested in. Customers are {\em single-minded} in a sense that each customer $j$ buys all items in $e_j$ if the sum of the prices does not exceed her budget (i.e. $b_j \geq \sum_{i \in e_j} p(i)$, where $p(i)$ indicates the price of item $i$), in which the seller gets $\sum_{i \in e_j} p(i)$ from the customer. Otherwise, the customer does not buy anything and the seller gets no profit from this customer. The goal of the seller is to set a nonnegative price to each item to maximize her profit from $m$ customers.

This problem was proposed by Guruswami et al.~\cite{GHKKKM05}, and has received much attention. Let $k$ be the maximum cardinality of any $e_i$. Approximability of this problem achieved by polynomial time algorithms for large $k$ and $n$ is relatively well-understood now. There is a polynomial time algorithm that guarantees $O(\min(k, (n \log n)^{1/2}))$ fraction of the optimal solution, while we cannot hope for an approximation ratio better than 
$\Omega(\min(k^{1 - \epsilon}, n^{1/2 - \epsilon}))$ for any $\epsilon > 0$ under the Exponential Time Hypothesis~\cite{CLN13}. 

The special case $k = 2$ has also been studied in many works separately.
The instance can be nicely represented by a graph, with vertices as items and edges as customers, so this problem is called the {\em Graph (Vertex) Pricing} problem. 
The fact that this case can be represented as a graph not only gives a theoretical simplification, 
but also makes the problem flexible to model other settings. 
For example, Lee et al.~\cite{LBATL07, LBATL08} independently suggested the same problem from the networking community, motivated by the study of pricing traffic between different levels of internet service providers under the presence of peering. 

The best known approximation algorithm for a general instance of Graph Pricing, which guarantees $\frac{1}{4}$ of the optimal solution, is given by Balcan and Blum~\cite{BB07} and Lee et al.~\cite{LBATL07} The algorithm is simple enough to state here. First, assign 0 to each vertex with probability half independently. For each remaining vertex $v$, assign the price which maximizes the profit between $v$ and its neighbors already assigned 0. 
This simple algorithm has been neither improved nor proved to be optimal. 
Graph Pricing is APX-hard~\cite{GHKKKM05}, but  the only strong hardness of approximation result rules out an approximation algorithm with a guarantee better than $\frac{1}{2}$~\cite{KKMS09} under the Unique Games Conjecture (UGC) (via reduction from Maximum Acyclic Subgraph).

The $\frac{1}{4}$-approximation algorithm is surprisingly simple and does not even rely on the power of a linear programming (LP) or semidefinite programming (SDP) relaxation. The efforts to exploit the power of LP relaxations to find a better approximation algorithm 
have produced positive results for special classes of graphs. 
Krauthgamer et al.~\cite{KMR11} studied the case where all budgets are the same (but the graph might have a self-loop), and proposed a $\frac{6+\sqrt{2}}{5 + \sqrt{2}} \approx 1.15$-approximation algorithm based on a LP relaxation. 
In general case, the standard LP is shown to have an integrality gap close to $\frac{1}{4}$~\cite{KKMS09}.
Therefore, it is natural to consider hierarchies of LP relaxations such as the Sherali-Adams hierarchy~\cite{SA90} (see \cite{CT12} for a general survey and \cite{GTW13, YZ14} for recent algorithmic results using the Sherali-Adams hierarchy). Especially, Chalermsook et al.~\cite{CKLN13} recently showed that there is a $\mathsf{FPTAS}$ when the graph has bounded treewidth, based on the Sherali-Adams hierarchy. 
However, the power of the Sherali-Adams hierarchy and SDP, as well as the inherent hardness of the problem, was not well-understood in general case. 

\subsection{Our Results}
In this work, we show that any polynomial time algorithm that guarantees a ratio better than $\frac{1}{4}$ must be powerful enough to refute the Unique Games Conjecture. 
\begin{thm}
Under the Unique Games Conjecture, for any $\epsilon > 0$, it is NP-hard to approximate Graph Pricing within a factor of $\frac{1}{4} + \epsilon$. 
\label{thm:main}
\end{thm}

By the results of Khot and Vishnoi~\cite{KV05} and Raghavendra and Steurer~\cite{RS09} that convert a hardness under the UGC to a SDP gap instance, our result unconditionally shows that even a SDP-based algorithm will not improve the performance of a simple algorithm. 
For the Sherali-Adams hierarchy, we prove that even polynomial rounds of the Sherali-Adams hierarchy has an integrality gap close to $\frac{1}{4}$. 

\begin{thm}
Fix $\epsilon > 0$. There exists $\delta > 0$ such that the integrality gap of $n^{\delta}$-rounds of the Sherali-Adams hierarchy for Graph Pricing is at most $\frac{1}{4} + \epsilon$.
\label{thm:gp}
\end{thm}

Our result is based on an interesting generalization of Max-Dicut, which we call Generalized Max-Dicut. It is parameterized by a positive integer $T \geq 1$. An instance consists of a directed graph $D = (V, A)$ and a label on each edge $l_A : E \rightarrow \left\{1, \dots, T \right\}$, where the goal is to assign to each vertex $v$ a label $l_V(v)$ from $\left\{0, \dots, T \right\}$ to maximize the number of satisfied edges --- each edge $(u, v)$ is satisfied if $l_V(u) = 0$ and $l_V(v) = l_A(u, v)$. 

This problem shares many properties with Graph Pricing, including a simple combinatorial $\frac{1}{4}$-approximation algorithm.
There is an approximation-preserving reduction from Generalized Max-Dicut($T$) on directed acyclic graphs (DAGs) to Graph Pricing for any $T$. We prove the following theorems that it is hard to improve upon this simple algorithm for large $T$ even on DAGs, which immediately imply Theorem~\ref{thm:main} and~\ref{thm:gp}.

\begin{thm}
Under the Unique Games Conjecture, 
it is NP-hard to approximate Generalized Max-Dicut($T$) on directed acyclic graphs within a factor of $\frac{1}{4} + O(\frac{1}{T^{1/4}})$.
\label{thm:ug-hard}
\end{thm}

\begin{thm}
Fix $T$ and $\epsilon > 0$. There exists $\delta > 0$ such that the integrality gap of $n^{\delta}$-rounds of the Sherali-Adams for Generalized Max-Dicut($T$) is at most $\frac{T+1}{4T}(1 + \epsilon)$.
Furthermore, the same result holds even when the graph is acyclic.
\label{thm:gmd}
\end{thm}

It is also interesting to compare the above results to other arity two Constraint Satisfaction Problems (CSPs), 
since whether the domain is Boolean (e.g. Max-Cut, Max-2SAT~\cite{GW95}) or not (e.g. 2-CSP with bounded domain~\cite{Hastad08}, Unique Games~\cite{CMM06}), SDP-based algorithms give a strictly better guarantee than LP-based or combinatorial algorithms. 
As discussed above, our result unconditionally says that a SDP-based algorithm cannot outperform a simple combinatorial algorithm for this arity two CSP (as $T$ increases).\footnote{ Formally, (approximation ratio of the SDP-based algorithm) / (approximation ratio of the best known combinatorial algorithm) = $1 + O(\frac{1}{T^{1/4}})$ for Generalized Max-Dicut. For Unique Games with $T$ labels, the SDP-based algorithm of Charikar et al.~\cite{CMM06}, which satisfies roughly $T^{-\epsilon / (2 - \epsilon)}$ fraction of constraints in an $(1 - \epsilon)$-satisfiable instance, performs better than the random assignment by any constant factor as $T$ increases.}

\subsection{Related Work and Our Techniques}

\paragraph{Formulation of Generalized Max-Dicut}
Our conceptual contribution is the introduction of Generalized Max-Dicut as a CSP that captures the complexity of Graph Pricing. It is inspired by the work of Khandekar et al.~\cite{KKMS09}, and our reduction is the almost same as their reduction from Max-Acyclic Subgraph (MAS) to Graph Pricing. 

In the natural formulation of Graph Pricing as a CSP, each vertex is assigned an (half-)integer price from $0$ to $B$ for the maximum budget $B$, and each customer becomes multiple constraints on two variables since the payoff linearly depends on the prices. It is shown in~\cite{KKMS09} that a half-integral optimal solution always exists for integral budgets, so this is a (almost) valid relaxation.
However, as each customer becomes multiple constraints with different payoffs, it seems hard to apply current techniques developed for well-studied CSPs to this formulation.

Khandekar et al.'s main idea was to use two well-known CSPs --- MAS for the hardness of approximation and Max-Dicut on directed acyclic graphs for the integrality gap of the standard LP. 
The former is harder to approximate, and the latter has the lower optimum.\footnote{
Under the Unique Games Conjecture, the best inapproximability ratio is 0.5 for MAS~\cite{GHMRC11} and 0.874 for Max-Dicut~\cite{Austrin10}. 
For the lower bound on integral optima, the maximum acyclic subgraph always has at least half of edges, while there is a directed acyclic graph where every directed cut cannot have more than $\frac{1}{4} + \epsilon$ fraction of edges for any $\epsilon > 0$.}
Generalized Max-Dicut seems to combine ingredients of both problems needed for Graph Pricing. It certainly inherits properties of Max-Dicut including low integral optima, but is much harder to approximate than Max-Dicut by Theorem~\ref{thm:ug-hard}.

\paragraph{Uniques Games-Hardness}
Proving hardness of Generalized Max-Dicut on general graphs is relatively straightforward --- proposing a dictatorship test with high completeness and low soundness, and plugging it into the {\em recipe} of Khot et al.~\cite{KKMO07} to deduce the hardness result. 
The dictatorship test is an instance of Generalized Max-Dicut with the set of vertices $\{ 0, ..., T \}^R$ (called {\em hypercube}) for some $R \in \mathbb{N}$. 
The main question in constructing a dictatorship test is how to sample $(x, y) \in \left\{0, \dots, T \right\}^2$, which induces a distribution on $\left\{0, \dots, T \right\}^2$. 
In Generalized Max-Dicut, $0$ is the only special label such that every directed edge is satisfied only if its tail is assigned 0. 
The simple combinatorial algorithm samples $0$ heavily --- the marginal distribution satisfies $\Pr[x = 0] \geq 0.5$, while the solution to the Sherali-Adams hierarchy constructed in Theorem~\ref{thm:gmd} treats 0 as other labels, having $\Pr[x = 0] = \frac{1}{T + 1}$. The latter distribution had a disadvantage that $x$ and $y$ are perfectly correlated --- the value of $x$ determines the value of $y$. 

To show the hardness based on the UGC (roughly equivalent to constructing a solution that {\em fools} SDP), we found that $\Pr[x = 0] = \frac{1}{T^{1/4}}$ is enough. In this case, we can ensure that the probability that dictators pass the test is large, while $x$ and $y$ behave almost independently. Based on the low correlation, we use the result of Mossel~\cite{Mossel10} to show low soundness. 

The resulting dictatorship test is not a DAG. 
To fix this problem, the final dictatorship test has the vertex set $V \times [T]^R$ for some DAG $D = (V, A)$.
For each edge $(u, v) \in A$, the above dictatorship test is performed so that each edge of the dictatorship test goes from the hypercube associated with $u$ to the one with $v$. This idea of keeping the dictatorship test acyclic is used in Svensson~\cite{Svensson12}, where he takes (the undirected version of) $D$ to be a complete graph. We take a nontrivial DAG found by Alon et al.~\cite{ABGLS07} where any directed cut has at most $\frac{1}{4} + o(1)$ fraction of edges. In the soundness case, if every hypercube is {\em pseudorandom}, the soundness analysis of an individual dictatorship test associated with each edge gives a rounding algorithm that finds a large directed cut in $D$, which contradicts the choice of $D$. 

This style of argument, composing the dictatorship test with a certain instance and solving this instance by the soundness analysis, resembles that of Raghavendra~\cite{Raghavendra08} for CSPs, Guruswami et al.~\cite{GHMRC11} for ordering CSPs, Kumar et al.~\cite{KMTV11} for strict CSPs, and Guruswami and Saket~\cite{GS10} for $k$-uniform $k$-partite Hypergraph Vertex Cover.
While they require the instance to have a good fractional solution (LP or SDP) but the low integral optimum, 
we only need the low integral optimum (of even a simpler problem) and our individual dictatorship test ensures completeness 
and part of soundness. We hope that this two-level technique --- constructing a simple dictatorship test for each edge and composing it with a certain instance with purely combinatorial properties --- makes it easier to bypass the barrier of finding a gap instance and prove hardness for many other problems, especially those with structured instances.

\paragraph{Sherali-Adams Gap}


On the integrality gap of Generalized Max-Dicut on a DAG, our work generalizes the work of Charikar et al.~\cite{CMM09}, which showed a similar result for Max-Cut, in several directions. 
The first obstacle is to find a DAG with a low integral optimum which is amenable to construct a good solution to the Sherali-Adams hierarchy. Previous works which obtained lower bounds for the Sherali-Adams hierarchy~\cite{ABLT06, dlVM07, CMM09} used $G(n, p)$, but $G(n, p)$ with an consistent orientation will not result in a low integral optimum. Instead, we show that sparsifying the aforementioned graph constructed in Alon et al.~\cite{ABGLS07}, which is already a DAG with a low integral optimum, gives other desired properties as well.


Given a set $S$ of $k$ vertices, we define a local distribution on the events $\left\{ l_V(v) = i \right\}_{v \in S, i \in T}$. One caveat of the above approach is that local distributions obtained might be inconsistent,
in a sense that $S$ and $S'$ might induce different marginal distributions on $S \cap S'$. Charikar et al.'s main idea is to embed them into $l_2$ and use hyperplane rounding to produce consistent ones. 
The most technical part of our work is to extend the hyperplane rounding to work for non-Boolean domains.
It is a complicated task in general, but we use the fact that the embedding is explicitly constructed for two adjacent vertices and it exhibits some symmetry, so that we can analyze the performance of our rounding. 
For $T = 1$, our result matches that of~\cite{CMM09}.

\subsection{Organization}
Section~\ref{sec:prelim} introduces problems and notations formally.
Section~\ref{sec:maxdicut_hardness} and Section~\ref{sec:gap} present Unique Games- hardness and Sherali-Adams integrality gaps of Generalized Max-Dicut respectively, which can be combined with the reduction in Section~\ref{sec:reduction} to give the same results for Graph Pricing.

\section {Preliminaries}
\label{sec:prelim}
For any positive integer $n$, let $[n] := \left\{0, 1, 2, \dots, n \right\}$ and $[n]^+ := \left\{1, 2, \dots, n \right\}$.
Given a sequence of numbers $a_1, ..., a_n$, let $\secondmax_j [a_j]$ be the second largest number among $a_j$'s. 

\paragraph{Graph Pricing} An instance of Graph Pricing consists of an undirected (possibly contain parallel edges) graph $G = (V, E)$ with budgets $b : E \rightarrow \mathbb{R}^+$ and weights $w : E \rightarrow \mathbb{R}^+$. Our goal is to find a pricing $p : V \rightarrow \mathbb{R}^+ \cup \left\{ 0 \right\}$ to maximize 
\[
\Val(p) := \sum_{e = (u, v) \in E} w(e) (p(u) + p(v)) \mathbb{I}[p(u) + p(v) \leq b(e)]
\]
where $\mathbb{I}[\cdot]$ is the indicator function. Let $\Opt(G, b, w) := \max_p \Val(p)$. 

\begin{rem}
This definition of Graph Pricing above coincides with General Graph Pricing defined in Khandekar et al.~\cite{KKMS09}. 
They presented an additional reduction from General Graph Pricing to Graph Pricing with no parallel edge and $w(e) = 1$. Throughout this paper, we use the definition above and allow weights and parallel edges for simplicity. In practice, weights can be naturally interpreted as the number of customers interested in the same pair.
\end{rem}

\begin{rem}
Another well-known pricing problem assumes that each customer will buy the cheapest item of her interest if she can afford it (i.e., $\Val(p) := 
\sum_{e = (u, v)} w(e) \min(p(u), p(v)) \mathbb{I}[\min(p(u), p(v)) \leq b(e)]$), which is called {\em unit-demand pricing}. Its approximability is similar to that of our single-minded pricing, including algorithms / hardness results for $k$-Hypergraph Pricing for large $k$~\cite{CLN13}, and a simple $\frac{1}{4}$-approximation algorithm for Graph Pricing $(k = 2)$. Indeed, Generalized Max-Dicut is also reducible to Unit-demand Graph Pricing and Theorem~\ref{thm:main} and~\ref{thm:gp} hold for it as well. We focus on Single-minded Graph Pricing here. 
\end{rem}

\paragraph{Generalized Max-Dicut} 
Fix a positive integer $T$. 
An instance of Generalized Max-Dicut($T$) consists of a digraph $D = (V, A)$ with a label $l_A : A \rightarrow [T]^+$ and a weight $w : A \rightarrow \mathbb{R}^+$ on each edge. Assume that  the sum of weights is normalized to 1. 
$(u, v)$ denotes the edge of $D$ from $u$ to $v$. We allow parallel edges from $u$ to $v$ if they have different labels (if parallel edges have the same label, simply merge them). 
Our goal is to find a labeling $l_V : V \rightarrow [T]$ (note vertices can be assigned 0, while edges are not) to maximize the weight of {\em satisfied} edges --- $(u, v)$ is satisfied when $l_V(u) = 0$ and $l_V(v) = l_A(u, v)$. 
Note than when $T = 1$, the problem becomes Max-Dicut. Given an instance $D = (V, A)$, $l_A$, and $w$, let 
$\Opt(G, l_A, w)$ be the maximum weight of edges satisfied by any labeling of vertices.
Given an assignment $l_V : V \rightarrow [T]$ to the vertices, let $\Val(l_V)$ be the weight of edges satisfied by $l_V$. Note that unlike Graph Pricing, the value of any assignment is normalized between 0 and 1.
The {\em normalized outdegree}, denoted by $\ndeg$, is defined to be $[\sum_u (\max_{(u, v) \in A} w(u, v))]^{-1}$. In unweighted instances (i.e. $w(e) = \frac{1}{|A|}$ for all $e$), $\ndeg \geq \frac{|A|}{|V|}$. 

\paragraph{Sherali-Adams Hierarchy}
In its most intuitive and redundant form, a feasible solution to the $r$-rounds of the Sherali-Adams hierarchy for a CSP with the domain $[q]$ consists of $\sum_{i = 1}^{r} \binom{n}{r} (q + 1)^r$ variables $\left\{ x_S(\alpha) \right\}$ for each subset of variables $S$ with cardinality at most $r$, and $\alpha \in [q]^S$. 
Each $x_S(\alpha)$ can be interpreted as the probability that the variables in $S$ are assigned $\alpha$. 
Therefore, it is required to satisfy the following natural conditions:
(1) $x_S(\alpha) \geq 0$ for all $S, \alpha$. 
(2) $\sum_{\alpha \in [q]^S} x_S(\alpha) = 1$ for all $S$.
(3) $\sum_{\alpha \in [q]^{S' \setminus S}} x_{S'}(\alpha \circ \beta) = x_S(\beta)$ for all $S \subseteq S', \beta \in [q]^S$, where $\alpha \circ \beta \in [q]^{S' \cup S}$ denote the joint assignment to the variables in $S'$.


The $r$-rounds of the Sherali-Adams hierarchy for Graph Pricing and Generalized Max-Dicut($T$) can be 
obtained by choosing an appropriate domain and an objective function, while using the constraints given above. 
For Graph Pricing, if we choose the domain to be $[B]$ where $B$ is the maximum budget, the objective function is the following.
\[
\sum_{e = (u, v)} w(e) \sum_{(i, j) \in [B]^2, i + j \leq b(u, v)} (i + j) \cdot x_{(u, v)}(i, j)
\]
Since $p(v)$ can be real, it is not clear whether this is a relaxation, even when the budgets are integers.
\cite{KKMS09} shows that there is a half-integral optimal solution.
The maximum budget $B$ can be exponentially big in the size of an instance,
and a standard trick is to consider only the powers of $(1 + \epsilon)$ as valid prices. It loses at most $\epsilon$ fraction of the optimum.
Our gap instance and proposed solution to the hierarchy have the marginal on each vertex supported by a constant number of prices, so they are applicable to any choice of the domain.

For Generalized Max-Dicut($T$), the domain is $[T]$, and the objective function is
\[
\sum_{(u, v) \in A} w(u, v) x_{(u, v)} (0, l_A(u, v)).
\]

Given an instance and a relaxation, we define the {\em integrality gap} to be the integral optimum divided by the value of the best solution to the relaxation. Since both our problems are maximization problems, it is at most 1 and a small number indicates a large gap. 

\section {Reduction from Generalized Max-Dicut to Graph Pricing}
\label{sec:reduction}
\begin{thm}
For any $T > 0$, there is a polynomial time reduction from an instance $(D = (V, A), l_A, \wgmd)$  of Generalized Max-Dicut($T$), where $D$ is acyclic and $\ndeg \geq \frac{1}{\epsilon}$, to an instance $(G, b, \wgp)$ of Graph Pricing such that $
\Opt(D, l_A, \wgmd) \leq \Opt(G, b, \wgp) \leq \Opt(D, l_A, \wgmd) + 3\epsilon.
$
\label{thm:reduction}
\end{thm}
\begin{proof}
Fix an instance $(D = (V, A), l_A, \wgmd)$ of Generalized Max-Dicut($T$) with $n = |V|$ and $m = |A|$.
Let $G$ be the underlying undirected graph of $D$. 
Our reduction from Generalized Max-Dicut on directed acyclic graphs to Graph Pricing is almost the same as the one in Khandekar et al.~\cite{KKMS09} with some simplification. Let $M$ be a large number which will be fixed later.

The resulting instance of Graph Pricing is based on the same graph $G$. 
Since $D$ is acyclic, there is an injective function $s : V \rightarrow [n]^+$
such that for each edge $(u, v) \in A$, $s(u) > s(v)$. For each edge $(u, v) \in A$, $b(u, v) = M^{Ts(v) + l_A(u, v) - 1}$ and $\wgp(u, v) = \frac{\wgmd(u, v)}{b(u, v)}$.

To avoid confusion, let $\Optgmd$, $\Valgmd$ denote $\Opt$, $\Val$ for Generalized Max-Dicut instances, and $\Optgp$ and $\Valgp$ for Graph Pricing instances. 
Fix a labeling $l_V : V \rightarrow [T]$. The corresponding {\em canonical solution} $p : V \rightarrow \mathbb{R^+} \cup \left\{ 0 \right\}$ defined by 
\[
p(v)
= \begin{cases}
M^{T s(v) + l_V(v) - 1} & \mbox {if } l_V(v) \neq 0 \\
0 & \mbox {otherwise}
\end{cases}
\]
gives $\Valgp(p) \geq \Valgmd(l_V)$ --- for each $(u, v) \in A$ satisfied by $l_V$, $p$ gets $p(v) \wgp(u, v) = \wgmd(u, v)$.
Therefore, $\Optgp(G, b, \wgp) \geq  \Optgmd(D, l_A, \wgmd)$. 
The following lemma shows that the converse is almost true. The proof is given in Appendix~\ref{sec:proof_reduction}.

\begin{lem}
[\cite{KKMS09}] For any $p$, $\Valgp(p) \leq \Optgmd(D, l_A, \wgmd) + \frac{1}{M} + 2\epsilon$.
\label{lem:reduction}
\end{lem}

Taking $M \geq \frac{1}{\epsilon}$ proves the theorem.
\end{proof}

\section {Approximability of Generalized Max-Dicut}
\label{sec:maxdicut_hardness}
Recall that Generalized Max-Dicut(1) is exactly the well-known Max-Dicut problem, 
which admits a 0.874-approximation algorithm~\cite{LLZ02} as any Max-2CSP over the Boolean domain. 
As $T$ increases, however, the best approximation ratio for Max-2CSP over the domain of size $T+1$ can be at most $O(\frac{\log T}{\sqrt{T}})$~\cite{Chan13}, so viewing it as a general Max-2CSP does not yield a constant-factor approximation algorithm. 

There is a simple $\frac{1}{4}$-approximation algorithm, similar to the one for Graph Pricing --- assign 0 to each vertex with probability half independently and assign nonzero values to the remaining vertices greedily. 
The proof is based on the fact that we can easily find the optimal solution once the set of vertices assigned 0 is given.
For small $T$, we can do a little better based on a standard LP relaxation. 
The proof is given in Appendix~\ref{sec:maxdicut_algo}.

\begin{thm}
There is a polynomial time approximation algorithm for Generalized Max-Dicut($T$) that guarantees 
$\frac{1}{4} + \Omega(\frac{1}{T})$ of the optimal solution.
\label{thm:gmd_algo}
\end{thm}

However, we prove that for large $T$, it is Unique Games-hard to improve the approximation ratio from $\frac{1}{4}$ to a better constant.

\begin{thm}
[Restatement of Theorem~\ref{thm:ug-hard}]
Under the Unique Games Conjecture, 
it is NP-hard to approximate Generalized Max-Dicut($T$) on directed acyclic graphs within a factor of $\frac{1}{4} + O(\frac{1}{T^{1/4}})$.
\end{thm}
Together with the reduction shown in Theorem~\ref{thm:reduction}, it immediately implies Theorem~\ref{thm:main} for Graph Pricing. Besides working on DAGs, the reduction also requires that $\ndeg$ be large, but it can be easily ensured by taking an Unique Games instance with large degree. See Appendix~\ref{sec:ugreduction} to see the full details. 

The theorem is proved by proposing a {\em dictatorship test} with high completeness and low soundness, 
combined with the standard technique to convert a dictatorship test to a hardness result based on the Unique Games Conjecture~\cite{KKMO07}. 
Constructing the dictatorship test has two components --- a simple dictatorship test based on correlation and Gaussian geometry, and composing it with a designated DAG. 
We present the dictatorship test here and defer the full reduction from Unique Games to Appendix~\ref{sec:ugreduction}. 

\subsection{Dictatorship Test}
We follow the notations in Mossel~\cite{Mossel10}.
Consider the hypercube $[T]^R$ where $[T] = \left\{0, 1, \dots, T \right\}$.
Let $\Omega_1 = \Omega_2 = [T]$. 
For $t \in [T]^+$, $\PP^t$ is a probability measure on $\Omega_1 \times \Omega_2$. Let $\PP$ be the marginal on $\Omega_i$ in $\PP^t$ (which does not depend on $t$ and $i$).
We want to ensure that 
$\PP(0) = \delta$, $\PP(j) = \frac{1 - \delta}{T}$ for $j \in [T]^+$ where $\delta = \frac{1}{T^{1/4}}$.
Let $\PP'$ be the distribution on $\Omega_1$ such that
$\PP'(0) = (\frac{1}{1 - \frac{1 - \delta}{T}})(\delta - \frac{1 - \delta}{T})$, $\PP'(j) = (\frac{1}{1 - \frac{1 - \delta}{T}})(\frac{1 - \delta}{T})$ (subtract $\frac{1 - \delta}{T}$ from $\PP(0)$ and renormalize). $\PP^t$ is defined by the following procedure to sample $(x, y)$. Sample $y$ according to $\PP$. If $y = t$, set $x = 0$. Otherwise, sample $x$ from $\PP'$ independently. It is easy to see that the marginal of both $x$ and $y$ is $\PP$. 
We show that $(x, y)$ are almost independent as $T$ increases. We define the correlation between two correlated spaces and prove the following lemma in Appendix~\ref{sec:proof_ughard}.
\begin{defi}
Given a distribution $\mathbf{Q}$ on $\Omega_1 \times \Omega_2$, we define the correlation $\rho(\Omega_1, \Omega_2; \mathbf{Q})$ by letting
\[
\rho(\Omega_1, \Omega_2; \mathbf{Q}) = \sup \left\{ \mathsf{Cov}[f, g] : f : \Omega_1 \rightarrow \mathbb{R}, g : \Omega_2 \rightarrow \mathbb{R}, \mathsf{Var}[f] = \mathsf{Var}[g] = 1 \right\}.
\]
\end{defi}
\begin{lem}
For any $t$, $\rho(\Omega_1, \Omega_2; \PP^t) \leq \sqrt{\frac{2}{T \delta}}$.
\label{lem:correlation}
\end{lem}

Another component of the dictatorship test is the directed acyclic graph $D = (V, A)$ of Alon et al.~\cite{ABGLS07},
where every directed cut has size at most $(\frac{1}{4} + o(1))|A|$. 
Fix a graph $D = (V, A)$ such that every dicut cuts at most $(\frac{1}{4} +  \frac{1}{T^{1/4}})|A|$ edges. 
Note that the size of this graph depends only on $T$.  
We now describe the dictatorship test. The prover is expected to provide $F_v : [T]^R \rightarrow [T]$ for each $v \in V$.
\begin{enumerate}
\itemsep= 0ex
\item Choose $(u, v) \in A$ and $t \in [T]^+$ uniformly at random.
\item For each $i \in [R]^+$, pick $(x_i, y_i)$ according to $\PP^t$.
\item Accept if $F_u(x) = 0$ and $F_v(y) = t$. 
\end{enumerate}

This dictatorship test can be naturally interpreted as an instance of Generalized Max-Dicut($T$) with the vertex set $V \times [T]^R$. The weight of edge $((u, x), (v, y))$ with label $t$ is equal to the probability that it is sampled, and a labeling $l : V \times [T]^R \mapsto [T]$ passes with probability $\Val(l)$ (by $F_v(x) = l(v, x)$). 

\subsection{Completeness and Soundness}
The $i$th dictator function is $D_i : [T]^R \rightarrow [T]$ given by $D_i(x_1, \dots, x_R) = x_i$.
The purpose of the above dictatorship test is to allow dictatorship functions to be accepted with high probability while penalizing functions {\em far from} any dictator. 
The following lemma for completeness is immediate from the test --- for any fixed $t$ and $i$, $\Pr[x_i = 0, y_i = t] = \Pr[y_i = t] = \frac{1 - \delta}{T}$. 

\begin{lem} [Completeness]
Suppose that for some $i$, $F_v = D_i$ for all $v \in V$. The above test accepts with probability $\frac{1 - \delta}{T}$.
\label{lem:completeness}
\end{lem}

For $v \in V$ and $t \in [T]$, let $F_{v, t} : [T]^R \rightarrow \left\{ 0, 1 \right\}$ be defined such that $F_{v, t}(x) = 1$ iff $F_v(x) = t$, and $\mu_{v, t}
:= \Pr[F_v(x) = t] = \E[F_{v, t}(x)]$ where $x \sim \PP$. 
For each $F_{v, t}$ and $i \in [R]^+$, we define the {\em influence} of the $i$th coordinate to measure the extent a function depends on the $i$th coordinate. 
\[
\mathsf{Inf}_i[F_{v, t}] := \E [ \mathsf{Var} [ F_{v, t} (X_1, \dots, X_R) | X_j, 1 \leq j \leq R, j \neq i]].
\]
We use a similarly defined {\em low-degree influence} $\Inf_i^{\leq d}$ for our soundness (see~\cite{Mossel10} for the definition). 
\begin{lem} [Soundness]
For large enough $T$, there exist $\tau$ and $d$ (depending on $T$) such that
if $\Inf_i^{\leq d}(F_{v, t}) \leq \tau$ for all $i \in [R]^+$, $t \in [T]$, and $v \in V$,
the probability of accepting is at most $\frac{1}{4T} + \frac{4}{T^{5/4}}$.
\label{lem:soundness}
\end{lem}

\begin{proof}
We use the following theorem of Mossel~\cite{Mossel10}.
\begin{thm}[Theorem 6.3 of~\cite{Mossel10}]
\label{thm:mossel}
Let $(\Omega_1 \times \Omega_2, \PP)$ be correlated spaces such that
the minimum nonzero probability of any atom in $\Omega_1 \times \Omega_2$ is at least $\alpha$ and such
that $\rho(\Omega_1, \Omega_2; \PP) \leq \rho$. 
Then for every $\epsilon > 0$ there exist $\tau, d$ depending on $\epsilon$ and $\alpha$ such that if $f : \Omega_1^R \rightarrow [0, 1], g : \Omega_2^R \rightarrow [0, 1]$ satisfy
$
\max(\Inf_i^{\leq d}[f], \Inf_i^{\leq d}[g]) \leq \tau
$
for all $i$, then
\[
\E_{(x, y) \in \PP^{\otimes R}} [f(x)g(y)] \leq \Gamma_{\rho}(\E_x[f], \E_y [g]) + \epsilon.
\]
\end{thm}
The probability of accepting is at most
\[
\E_{(u, v) \in A} [ \E_{t \in [T]^+} [ \E_{(x, y) \sim (\PP^t)^{\otimes R}} [F_{u, 0}(x) F_{v, t}(y)] ] ] \\
\leq
\E_{(u, v) \in A} [ \E_{t \in [T]^+} [
 \Gamma_{\rho}(\mu_{u, 0}, \mu_{v, t}) + \frac{1}{T^{5/4}} ] ] 
\]
where the inequality follows from Theorem~\ref{thm:mossel} (set $\epsilon \leftarrow \frac{1}{T^{5/4}}$ and $\alpha = \Theta(\frac{1}{T^2})$). 
The following lemma, whose proof is given in Appendix~\ref{sec:gaussian}, shows that it is at most 
\[\E_{(u, v) \in A} [ 
\Gamma_{\rho}(\mu_{u, 0}, \frac{1 - \mu_{v, 0}}{T}) ] + \frac{1}{T^{5/4}}.
\]
\begin{lem}
Fix $\rho, a \in (0, 1)$. The function $f(b) := \Gamma_{\rho}(a, b)$ is concave. 
\label{lem:gauss_concave}
\end{lem}
The following lemma, whose proof is again given in Appendix~\ref{sec:gaussian}, shows that it is at most 
\[
\E_{(u, v) \in A} [ 
\mu_{u, 0} \frac{(1 - \mu_{v, 0})}{T} + \frac{2}{T^{5/4}}] + \frac{1}{T^{5/4}}
=
\frac{1}{T} \E_{(u, v) \in A} [\mu_{u, 0}(1 - \mu_{v, 0})] + \frac{3}{T^{5/4}}.
\]
\begin{lem}
For large enough $T$ and $\delta = \frac{1}{T^{1/4}}$, the following holds. 
For any $a \in [0, 1], b \in [0, \frac{1}{T}]$ and $\rho \in (0, \sqrt{\frac{2}{T\delta}})$, 
$\Gamma_{\rho}(a, b) \leq ab + \frac{2}{T^{5/4}}$.
\label{lem:gauss_bound}
\end{lem}
Given $\{ \mu_{v, 0} \}_{v \in V}$, imagine the rounding algorithm which puts $v \in S$ with probability $\mu_{v, 0}$ independently. 
The expected fraction of edges from $S$ to $V \setminus S$ is $\E_{(u, v) \in A} [\mu_{u, 0}(1 - \mu_{v, 0})]$, which is at most the fractional size of maximum dicut of $D$. Since we took $D$ to satisfy that $\E_{(u, v) \in A} [\mu_{u, 0}(1 - \mu_{v, 0})] \leq \frac{1}{4} + \frac{1}{T^{1/4}}$, the probability of accepting is at most $\frac{1}{4T} + \frac{4}{T^{5/4}}$ as desired.
Note that the probabilities of accepting in completeness and soundness differ by a factor of 
$
\frac{\frac{1}{4T} + \frac{4}{T^{5/4}} }{\frac{1 - \delta}{T}}
= \frac{ \frac{1}{4} + \frac{4}{T^{1/4}} } {1 - \frac{1}{T^{1/4}}}
= \frac{1}{4} + O(\frac{1}{T^{1/4}}).
$
\end{proof}



\section {Integrality Gaps for Generalized Max-Dicut}
\label{sec:gap}
Fix a positive integer $T$ and $\epsilon \in (0, \frac{1}{100})$. 
We present an instance of Generalized Max-Dicut($T$) $(D = (V, A), l_A)$ (we only deal with unweighted instances in this section and omit $w$) such that $D$ is acyclic, $|V| \leq \epsilon |A|$ (so that $\ndeg \geq \frac{1}{\epsilon})$, and a solution to $n^{\delta}$-rounds of the Sherali-Adams hierarchy such that the integrality gap is at most $\frac{T+1}{4T}(1 + \epsilon)$.
This result almost matches a simple $\frac{1}{4}$-approximation algorithm. 

Through the reduction given in Theorem~\ref{thm:reduction}, we also prove Theorem~\ref{thm:gp} --- a bad integral solution is guaranteed by the reduction, a good solution to the Sherali-Adams hierarchy is obtained by the mapping $l_V(u) = i$ to $p(u) = M^{Ts(u) + i - 1}$ (if $i \neq 0$) or 0 (otherwise). 
The budget in the resulting instance is an integer exponential in the size of instances, and our gap works even for a strong linear programming hierarchy where there is a variable for each vertex $v$ and an integer price $i$. 

The rest of this section is devoted to the proof of Theorem~\ref{thm:gmd}. 

\subsection{Obtaining a Good Instance}
\label{subsec:instance}

Our graph $D$ is obtained by randomly sparsifying the graph $D_* = (V, A_*)$ constructed in Alon et al.~\cite{ABGLS07}, followed by an appropriate postprocessing. 
$D_*$ is a directed acyclic graph with $n$ vertices and $m_* = \Theta(n^{\frac{5}{3}})$ edges. Its underlying undirected graph $G_* = (V, E_*)$ is a simple graph with the same number of vertices and edges, with the maximum degree $\Delta_* = \Theta(n^\frac{2}{3})$. Actually, $V = [n]^+$ and $(u, v) \in E$ only if $|u - v| \leq r$ where $r := \Theta(n^{\frac{2}{3}})$. It has the property that any directed cut has size at most $\frac{m_*}{4} + o(m_*)$ edges. 

The first version of $D = (V, A)$ is constructed as the following. $V := V_* = [n]^+$, and for each edge $(u, v) \in A_*$, put $(u, v) \in A$ with probability $p := \frac{\Delta}{\Delta_*}$ for some $\Delta$ to be fixed later. Let $G = (V, E)$ be the underlying undirected graph of $D$. 
$l_A$ is obtained by assigning each $l(u, v)$ a random number from $[T]^+$. 

Like previous integrality gap constructions for Max-Cut and Min-Vertex Cover (e.g.~\cite{ABLT06, dlVM07, STT07, CMM09}) , $D$ must be postprocessed to be amenable to have a Sherali-Adams solution with a large value.
Intuitively, we need to have the underlying undirected graph $G$ {\em locally sparse} --- if we look at a neighborhood of a certain vertex, the graph almost looks like a tree. We use the notion of \cite{CMM09} to measure how locally sparse the graph is. 

\begin{defi}
We say that $G'$ is $l$-path decomposable if every 2-connected subgraph $H$ of $G'$ contains a path of length $l$ such that every vertex of the path has degree 2 in $H$.
\end{defi}

The first version of the instance already has $\Opt(D, l_A) \approx \frac{1}{4T}$ with high probability. In order to make the instance locally sparse, we additionally need to remove some of the edges, but the fraction of removed edges is so small that it does not affect $\Opt(D, l_A)$ too much. 
As a result, we get the following theorem. The proof is given in Appendix~\ref{sec:proof_gap}.

\begin{thm}
Given $T$ and $\epsilon, \mu > 0$, there exist constants $\Delta, \delta$ and $l = \Theta(\log n)$ (all constants depending on $T$ and $\epsilon, \mu$) such that there is an instance of Generalized Max-Dicut($T$) $(D, l_A)$ with the underlying undirected graph $G$ with the following properties.
\begin{itemize}
\itemsep=0ex
\item Acyclicity: $D$ is a DAG.
\item Low integral optimum: $\Opt(D, l_A) \leq \frac{1 + \epsilon}{4T}$.
\item Almost regularity: Maximum degree of $G$ is at most $2\Delta$, and $G$ has at least $\Omega(\Delta n)$ edges.
\item Local sparsity: For $k \leq n^{\delta}$, every induced subgraph of $G$ on $(2\Delta)^lk$ vertices is $l$-path decomposable.
\item Large noise: For $k \leq n^{\delta}$, $(1 - \mu)^{l/10} \leq \frac{\mu}{5k}$.
\end{itemize}
\label{thm:goodinstance}
\end{thm}
The last condition, large noise, is needed to ensure that in a LP solution, 
even though adjacent vertices are very correlated to give a large value, far away vertices behave almost independently. The meaning of each condition will be elaborated in later sections.


\subsection{Constructing (Inconsistent) Local Distributions}
\label{subsec:inconsistent_local}
Let $D = (V, A)$, $l_A$, and $G = (V, E)$ be the instance of Generalized Max-Dicut($T$) and its underlying undirected graph constructed as above. In this subsection, given a set of $k \leq n^{\delta}$ vertices $S = \left\{v_1, \dots, v_k \right\}$ we give a distribution on events 
\[ \left\{ l_V(v_1) = x_1, \dots, l_V(v_k) = x_k \right\}_{x_1, \dots, x_k \in [T]}. \]

The local distributions we construct in this subsection are not consistent; for different sets $S$ and $S'$, the marginal distribution on $S \cap S'$ from the distribution on $S$ can be different from the same marginal from the distribution on $S'$ (albeit they are close). This problem is fixed in the next subsection. 

Let $d(u, v)$ be the shortest distance between $u$ and $v$ in $G$ and $V' \subseteq V$ be the set of vertices whose shortest distance to $S$ is at most $l$. Let $G'$ and $D'$ be the subgraph of $G$ and $D$ induced on $V'$, respectively. Since $|V'| \leq (2\Delta)^l k$, $G'$ is $l$-path decomposable by Theorem~\ref{thm:goodinstance}.
Note that if $d(u, v) < l$, $d(u, v)$ is also the shortest distance between $u$ and $v$ in $G'$.
By the definition, a $l$-path decomposable graph does not have a cycle of length $l$, so if $d(u, v) < \frac{l}{2}$, the shortest path between $u$ and $v$ must be unique.

We begin by establishing a fact that when $G'$ is path-decomposable (intuitively looks similar to a tree), there is a distribution on the partitions of $V$ (i.e. multicuts) such that close vertices are unlikely to be separated but far vertices are likely to be separated. If $G'$ is a tree, it is obtained by deleting each edge independently with probability $\mu$. 
The noise parameter $\mu$ will be fixed later depending only on $T$ and $\epsilon$, so is asymptotically greater than $\frac{1}{l} = O(\frac{1}{\log n})$.

\begin{thm}
[\cite{CMM10}]
Suppose $G' = (V, E)$ is an $l$-path decomposable graph. Let $L = \lfloor l / 9 \rfloor; \mu \in [1/L, 1]$. Then there exists a probabilistic distribution of multicuts of $G'$ (or in other words random partition of $G'$ in pieces) such that the following properties hold. For every two vertices $u$ and $v$,
\begin{enumerate}
\itemsep=0ex
\item If $d(u, v) \leq L$, then the probability that $u$ and $v$ are separated by the multicut (i.e. lie in different parts) equals $1 - (1 - \mu)^{d(u, v)}$; moreover, if $u$ and $v$ lie in the same part, then the unique shortest path between $u$ and $v$ also lies in that part. 

\item If $d(u, v) > L$, then the probability that $u$ and $v$ are separated by the multicut is at least $1 - (1 - \mu)^L$.

\item Every piece of the multicut partition is a tree.
\end{enumerate}
\label{thm:partition}
\end{thm}

Based on this random partitioning, we define the distribution on the vertices in $S$ (actually in $V'$). 
For each piece which is a tree, pick an arbitrary vertex $v$ in the tree, choose $l_V(v)$ uniformly at random, 
and propagate this label to {\em weakly satisfy} every edge in the tree --- an undirected edge $(u', v') \in E$ (swap $u'$ and $v'$ if necessary to assume $(u', v') \in A$) is weakly satisfied when $l_V(v') - l_V(u')= l_A(u', v')$ over $\ZZ_{T + 1}$. Note that this definition is necessary for the original definition of satisfaction, but not sufficient. 

It is clear that the choice of root in each tree does not matter, and the marginal distribution of each $l_V(v)$ is uniform on $[T]$. For vertices $u$ and $v$ with $d(u, v) \leq L$, we say that label $i$ for $u$ and $i'$ for $v$ {\em match} if $l_V(u) = i, l_V(v) = i'$ can be extended to weakly satisfy every edge on the unique shortest path between $u$ and $v$ (there are $T + 1$ such pairs). If $u$ and $v$ are close, $l_V(u)$ and $l_V(v)$ will be correlated in a sense that if $i$ and $i'$ match, $l_V(u) = i$ almost implies $l_V(v) = i'$, while it is not the case when $u$ and $v$ are far apart. The following corollary formalizes this intuition. The proof is in Appendix~\ref{subsec:proof_distribution}.


\begin{cor}
Suppose $G' = (V', E')$ is an $l$-path decomposable graph. Let $L = \lfloor l / 9 \rfloor$; $\mu \in [1 / L, 1]$. Then there exists a random mapping $r : V' \rightarrow [T]$ such that
\begin{enumerate}
\itemsep=0ex
\item If $d := d(u, v) \leq L$ then 
\[
\Pr[r(u) = i, r(v) = i'] = 
\begin{cases}
\frac{(1 - \mu)^d}{(T + 1)} + \frac{1 - (1 - \mu)^d}{(T + 1)^2} & \mbox{if } i \mbox{ and } i' \mbox{ match} \\
\frac{1 - (1 - \mu)^d}{(T + 1)^2} & \mbox{otherwise}
\end{cases}
\]
\item If $d > L$ then $\frac{1 - (1 - \mu)^L}{(T + 1)^2} \leq \Pr[r(u) = i, r(v) = i'] \leq \frac{1 - (1 - \mu)^L}{(T + 1)^2} + \frac{(1 - \mu)^L}{T + 1}$ for any $i, i' \in [T]$.
\end{enumerate}
\label{cor:distribution}
\end{cor}

\begin{defi}
For any vertices $u \neq v$ and $i, i' \in [T]$, let 
$\rho(u(i), v(i')) := \Pr[r(u) = i, r(v) = i']$ if $d(u, v) \leq L$, or $\frac{1}{(T+1)^2}$ otherwise. 
$\rho(v(i), v(i)) := \frac{1}{T+1}$ and $\rho(v(i), v(i')) := 0$ for $i \neq i'$.
Since the shortest path between $u$ and $v$ is unique when $d(u, v) \leq L$, $\rho$ is uniquely defined given $G$, $D$, $l_A$ and does not depend on $S,$ $V'$, $G'$, $D'$ which induce a local distribution.
\end{defi}

\begin{defi}
Fix a set of $k$ vertices $S = \left\{ v_1, \dots, v_k \right\}$. 
For any vertex $u, v \in S$ and $i, j \in [T]$, let 
$\nu_S(u(i), v(i')) := \Pr[x(u) = i, x(v) = i']$ in the local distribution on $S$ defined by $r$ in Corollary~\ref{cor:distribution}. 
\end{defi}

\subsection{Geometric Embedding and Rounding}
\label{subsec:embedding_rounding}
In this subsection, we still fix a set of $k$ vertices $S = \left\{ v_1, \dots, v_k \right\}$ and produce a distribution on the events $\left\{ l_V(v_1) = x_1, \dots, l_V(v_k) = x_k \right\}_{x_1, \dots, x_k \in [T]}$. 
The difference from the last subsection is that the resulting distributions become consistent --- the marginal distribution on $S \cap S'$ does not depend on the choice of its superset ($S$ or $S'$) that is used to obtain a larger local distribution. 

\subsubsection{Embedding}
Consider $\rho$ and $\nu_S$ defined in the last subsection. 
$\rho$ and $\nu_S$ both capture the pairwise distribution between the events $\left\{ l_V(v) = x \right\}_{v \in S, x \in [T]}$, but each of them has its own defects. $\nu_S$ depends on the choice of $S$, so does not yield consistent local distributions. $\rho$ does not depend on $S$, but for far vertices, Corollary~\ref{cor:distribution} does not guarantee any local distribution consistent with it. However, they are close in a sense --- they are identical when $d(u, v) \leq L$ and differ by at most $\frac{(1 - \mu)^L}{T + 1}$ otherwise. 

The main idea of Charikar et al.~\cite{CMM09} is to 
interpret $\rho$ and $\nu_S$ as pairwise distances between events
and embed $\rho$ to $l_2$ with small error. It is based on the fact that $\rho$ and $\nu_S$ are close for any $S$ and $\nu_S$ is readily embeddable to $l_2$. 
Since the embedding into $l_2$ is uniquely defined by the pairwise distances and $\rho$ does not depend on the choice of $S$, geometric rounding schemes based on the embedding yield consistent local distributions.
Let $v(i)$ be the vector corresponding to the event $l_V(v) = i$. Our goal is to construct $k(T+1)$ vectors $\left\{ v(i) \right\}_{v \in S, i \in [T]}$ such that $u(i) \cdot v(i') \approx \rho(u(i), v(i'))$. 
Following the above intuition, the following lemma says that this embedding is possible with error depending on $\mu$. The proof is given in Appendix~\ref{subsec:proof_embedding}.

\begin{lem}
There exist $k(T+1)$ vectors $\left\{ v(i) \right\}_{v \in S, i \in [T]}$ such that $\| v(i) \|_2^2 = \mu + \frac{1}{T + 1}$ and $u(i) \cdot v(i') = \frac{\mu}{2} + \rho(u(i), v(i'))$. 
\label{lem:embedding}
\end{lem}

\subsubsection{Rounding and Analyzing adjacent vertices}
Given $k(T+1)$ vectors $\left\{ v(i) \right\}_{v \in S, i \in [T]}$, 
our rounding scheme is one of the most natural ways to choose one out of $(T+1)$ vectors --- take a random Gaussian vector $g$ and for each vertex $v$, set $l_V(v) = i$ such that $v(i) \cdot g$ is the maximum over all $i$.
Since the inner products of these vectors depend only on $\rho$ (which does not depend on the choice of $S$), it gives a consistent local distribution.

Fix adjacent vertices $v$ and $u$ (without loss of generality assume $(u, v) \in A$). It only remains to show that $\Pr[l_V(u) = 0, l_V(v) = l_A(u, v)] \approx \frac{1}{T + 1}$. For any pair of adjacent vertices, we can write $2(T+1)$ vectors explicitly.
They are just two sets of $T+1$ orthonormal vectors, very closely correlated --- there are $T+1$ pairs $(u(i), v(i'))$, $i' - i = l_A(u, v)$ in $\ZZ_{T+1}$, such that $u(i) \approx v(i')$. 
With this symmetric structure and a suitable choice of the noise parameter $\mu$, we can analyze the performance of our rounding. The proof is given in Appendix~\ref{subsec:proof_rounding}.

\begin{lem}
There exists $\mu$ depending on $T$ and $\epsilon$ such that, in the above rounding scheme, the probability that $l_V(u) = 0$ and $l_V(v) = l_A(u, v)$ is at most $\frac{1 - 12\epsilon}{T + 1}$.
\label{lem:rounding}
\end{lem}

This finishes the construction of a solution to the $n^{\delta}$-rounds of the Sherali-Adams hierarchy with value $\frac{1 - 12\epsilon}{T + 1}$. Since $\Opt(V, l_A) \leq \frac{1 + \epsilon}{4T}$ by Theorem~\ref{thm:goodinstance}, it proves Theorem~\ref{thm:gmd} and Theorem~\ref{thm:gp}.

\paragraph{Acknowledgements.} The author would like to thank Venkat Guruswami and Seung Woo Shin for helpful discussions.

\bibliographystyle{abbrv}
\bibliography{../mybib}

\appendix

\section{Proof of Lemma in the Reduction}
\label{sec:proof_reduction}
\begin{lem}
[\cite{KKMS09}, Restatement of~\ref{lem:reduction}] For any $p$, $\Valgp(p) \leq \Optgmd(D, l_A, \wgmd) + \frac{1}{M} + 2\epsilon$.
\end{lem}
\begin{proof}
Given $p$, we define the {\em principal part} of $\Valgp(p)$ as
\[
\sum_{(u, v) \in A} \wgp(u, v) p(v) \mathbb{I}[p(u) + p(v) \leq b(u, v)].
\]
Note that for each directed edge, only the price of its head contributes. 

We first bound the principal part of $\Valgp(p)$. 
For a vertex $v$, the only edges where $\wgp(u, v) p(v) > \frac{\wgmd(u, v)}{M}$ satisfy $M^{Ts(v) + l_A(u, v) - 2} < p(v) \leq M^{Ts(v) + l_A(u, v) - 1}$. If there is such an edge, let $l'_V(v) = l_A(u, v)$. Otherwise, let $l'_V(v) = 0$. 
Fix an edge $(u, v)$ where $\wgp(u, v)p(v) > \frac{\wgmd(u, v)}{M}$. $l'_V(v) = l_A(u, v)$ by above. If $l'_V(u) \neq 0$, it means $p(u) > M^{Ts(u) - 1} \geq M^{Ts(v) + T - 1} \geq b(u, v)$, so $(u, v)$ contributes 0 to the principal part of $\Valgp(p)$. 
Therefore, for each edge $(u, v)$ that contributes more than $\frac{1}{M}$ to the principal part of $\Valgp(p)$, $l'_V$ satisfies $(u, v)$. Therefore, the principal part of $\Valgp(p)$ is at most $\Optgmd(D, l_A, \wgmd) + \frac{1}{M}$.

For the non-principal part of $\Valgp(p)$, for each vertex $u$, and we bound 
\[
\sum_{(u, v) \in A} \wgp(u, v) p(u) \mathbb{I}[p(u) + p(v) \leq b(u, v)] 
\leq \sum_{(u, v) \in A, p(u) \leq b(u, v)} \wgmd(u, v) \frac{p(u)}{b(u, v)}.
\]
Note that all edges $(u, v)$ have different $b(u, v)$, and any two differ by at least a factor of $M$. Let $w_u := \max_{(u, v) \in A} \wgmd(u, v)$. 
Therefore, the right hand side can be bounded by $w_u(1 + \frac{1}{M} + \frac{1}{M^2} + \dots) \leq 2w_u$, where 
\[
\sum_u w_u = \frac{1}{\ndeg} \leq \epsilon.
\]
This shows that the non-principal part of $\Valgp(p)$ is at most $2\epsilon$, proving the lemma.
\end{proof}

\section{Details of the Integrality Gap}
\label{sec:proof_gap}
\subsection{Obtaining a Good Instance}
\label{subsec:proof_instance}
In this subsection, we prove the following theorem.
\begin{thm}
[Restatement of Theorem~\ref{thm:goodinstance}]
Given $T$ and $\epsilon, \mu > 0$, there exist constants $\Delta, \delta$ and $l = \Theta(\log n)$ (all constants depending on $T$ and $\epsilon, \mu$) such that there is an instance of Generalized Max-Dicut($T$) $(D, l_A)$ with the underlying undirected graph $G$ with the following properties.
\begin{itemize}
\itemsep=0ex
\item Acyclicity: $D$ is a DAG.
\item Low Integral Optimum: $\Opt(D, l_A) \leq \frac{1 + \epsilon}{4T}$.
\item Almost regularity: Maximum degree of $G$ is at most $2\Delta$, and $G$ has at least $\Omega(\Delta n)$ edges.
\item Local Sparsity: For $k \leq n^{\delta}$, every induced subgraph of $G$ on $(2\Delta)^lk$ vertices is $l$-path decomposable.
\item Large noise: For $k \leq n^{\delta}$, $(1 - \mu)^{l/10} \leq \frac{\mu}{5k}$.
\end{itemize}
\end{thm}
\begin{proof}
As mentioned in Section~\ref{subsec:instance}, our graph $D$ is obtained by randomly sparsifying the graph $D_* = (V, A_*)$ constructed in Alon et al.~\cite{ABGLS07} after an appropriate postprocessing. 
$D_*$ is a directed acyclic graph with $n$ vertices and $m_* = \Theta(n^{\frac{5}{3}})$ edges. Its underlying undirected graph $G_* = (V, E_*)$ is a simple graph with the same number of vertices and edges, with the maximum degree $\Delta_* = \Theta(n^\frac{2}{3})$. Actually, $V = [n]^+$ and $(u, v) \in E$ only if $|u - v| \leq r$ where $r := \Theta(n^{\frac{2}{3}})$. It has the property that any directed cut has size at most $\frac{m_*}{4} + o(m_*)$ edges. 

The first version of $D = (V, A)$ is constructed as the following. $V := V_* = [n]^+$, and for each edge $(u, v) \in A_*$, put $(u, v) \in A$ with probability $p := \frac{\Delta}{\Delta_*}$ for some $\Delta$ to be fixed later. Let $G = (V, E)$ be the underlying undirected graph of $V$. 
$l_A$ is obtained by assigning each $l(u, v)$ a random number uniformly sampled from $[T]^+$. 

\paragraph{Integral Solution}
The following lemma shows that if $\Delta$ is big enough, $\Opt(D, l_A)$ is close to $\frac{1}{4T}$. 

\begin{lem}
If $G$ satisfies the above four properties and $\Delta = \Omega(\frac{T \log T}{\epsilon^2})$, then $D$ and $l_A$ obtained by the above process satisfies $\Opt(D, l_A) \leq \frac{1 + 4\epsilon}{4T}$ with high probability.
\label{lem:integral2}
\end{lem}
\begin{proof}
Fix one assignment $l_V : V \rightarrow [T]$. For any edge $(u, v) \in A_*$ call it a {\em candidate} when $l_V(u) = 0, l_V(v) \neq 0$. 
Note that the number of candidate edges is at most the cardinality of the maximum directed cut of $D_*$, which is at most 
$\frac{1 + o(1)}{4}m_*$.

For each candidate edge $(u, v)$, the probability that $(u, v) \in A$ with $l_A(u, v) = l_V(v)$ is $\frac{1}{T}$. Therefore, the expected number of satisfied edges is at most $\frac{(1 + o(1)) \Delta m_*}{4 \Delta_* T}$.
By Chernoff bound, the probability that it is bigger than $\frac{(1 + \epsilon) p m_*}{4 T}$ is bounded by $\exp(-\Omega(\frac{\epsilon^2 p m_*}{T})) = \exp(-\Omega(\frac{\epsilon^2 \Delta n}{T}))$. By taking union bound over $(T+1)^n$ different $l_V$'s, the probability that there exists an assignment with more than $\frac{(1 + \epsilon) p m_*}{4T}$ satisfied edges is at most 
\[
\exp(-\Omega(\frac{\epsilon^2 \Delta n}{T})) * \exp(n \log (T+1)) \leq n^{-1} 
\]
for $\Delta := \Omega(\frac{T \log T}{\epsilon^2})$. 
Similarly, we can conclude that $|A| \geq (1 - \epsilon)m_*p$ with high probability. Therefore, $\Opt(D, l_A)$ is at most $\frac{(1 + \epsilon)}{4T(1 - \epsilon)} \leq \frac{1 + 4\epsilon}{4T}$ with high probability.
\end{proof}

The above lemma is the only place where it is desirable to have large $|A| = |E|$. 
For the rest of this subsection, we are going to delete some edges of $D$ (and $G$) to satisfy desired properties.
Note that in any case, the number of edges deleted is much less than $\epsilon p m_* $ so that each deletion does not hurt the above lemma. 

\paragraph{Maximum Degree Control}
Since the maximum degree in $G_*$ is $\Delta_*$, expected degree of each vertex $v \in V$ in $G$ is at most $p\Delta_* = \Delta$.
Call a vertex $v \in V$ {\em bad} if it has degree more than $2\Delta$ in $G$, and call an edge $(u, v) \in E$ bad if either $u$ or $v$ is bad.
Fix an edge $(u, v)$. The probability that $(u, v)$ becomes bad given $(u, v) \in E$ is at most $2\exp(-\frac{\Delta}{4})$. 
The expected number of bad edges is at most $2\exp(-\frac{\Delta}{4})pm_*$, and by Markov's inequality, with probability at least half, the number of bad edges is at most $4 \exp(-\frac{\Delta}{3}) pm_*$. 

Deleting all bad edges guarantees that the maximum degree of $G$ is at most $2\Delta$, and with probability at least half, we delete only $4 \exp(-\frac{\Delta}{3}) pm_*$ edges, which is much smaller than $\epsilon p m_*$ since $\Delta = \Omega(\frac{1}{\epsilon^2})$.

\paragraph{Girth Control}
The expected number of cycles of length $i$ is bounded by
\[
n(2r)^{i - 1} p^i =
n(2r)^{i - 1} (\frac{\Delta}{\Delta^*})^i \leq 
\frac{n (C\Delta)^i}{\Delta_*}
\]
for some absolute constant $C$. When $i = O(\frac{\log n}{\log \Delta})$ the above quantity becomes less than $n^{0.5}$. Assume $l = O(\frac{\log n}{\log \Delta})$ (it will be fixed even smaller than that later). Summing over $i = 4, \dots, l$ ensures that the expected number of cycles of length up to $l$ is at most $O(n^{0.6})$, and it is less than $O(n^{0.7})$ with high probability. Removing one edge for each cycle of length up to $l$ ensures that $G$ has girth at least $l$. 

\paragraph{Local Sparsity Control}

Let $\eta = \frac{1}{3l}$ for some $l$ fixed later. 
We want to show that there exists $\gamma > 0$ such that every subgraph $G'$ of $G$ induced on $t \leq n^{\gamma}$ vertices have only $(1 + \eta)t$ edges. 

For $4 \leq t \leq 1/\eta$, we count the number of connected subgraphs of $G_*$ with $t$ vertices and $t+1$ edges.
\begin{lem}
The number of connected subgraphs of $G_*$ with $t$ vertices and $t+1$ edges is bounded by $2nt^2 \Delta_*^{t-1}$.
\end{lem}
\begin{proof}
The only possible degree sequences for such subgraphs are $(4, 2, 2, 2, \dots)$ or $(3, 3, 2, 2, \dots)$. Assume that it is $(4, 2, 2, 2, \dots)$. Let $v$ be the vertex with degree 4. There is a sequence of $t + 2$ vertices 
$(v, \dots, v, \dots, v)$ representing an Eulerian tour (not necessarily unique). The number of such sequences is bounded by $nt\Delta_*^{t-1}$ ($n$ for guessing $v$, $t$ for guessing where $v$ occurs in the middle of the sequence, $\Delta_*^{t-1}$ for the other vertices). 

Assume that the degree sequence is $(3, 3, 2, 2, \dots)$, and $u, v$ be the vertices of degree 3. Take a sequence of $t+2$ vertices representing an Eulerian path from $u$ to $v$ (either $(u, \dots, u, \dots, v, \dots, v)$ or 
$(u, \dots, v, \dots, u, \dots, v)$). The number of such sequences is bounded by $nt^2 \Delta_*^{t-1}$ ($n$ for guessing $u$, $t^2$ for guessing positions of $u$ and $v$ in the middle of the sequence, $\Delta_*^{t-1}$ for the other vertices including $v$). 
\end{proof}

Therefore, the probability that there exists a subgraph of $G$ with $t$ vertices and $t+1$ edges for $4 \leq t \leq 1 / \eta = 3l$ is
\[
\sum_{t=4}^{3l} 2nt^2 \Delta_*^{t-1}p^{t+1} = \sum_{t=4}^{3l} \frac{2nt^2\Delta^{t+1}}{\Delta_*^2} 
\leq \frac{n}{\Delta_*^2} (9l)^2 \Delta^{3l+1} \leq n^{-0.1}
\]
for $l = O(\log n / \log \Delta)$, since $\frac{n}{\Delta_*^2} = O(n^{-\frac{1}{3}})$. 

For $t > 1 / \eta = 3l$, we count the number of subgraphs of $G_*$ with $t$ vertices and $(1 + \eta)t$ edges.
It is upper bounded by (the number of connected subtrees on $t$ vertices) * (the number of possibilities to choose other $\eta t + 1$ edges out of $\binom{t}{2}$ pairs). The number of unlabeled rooted trees on $t$ vertices is $C \alpha^t$ for some constants $C$ and $\alpha$~\cite{Otter48}, so the number of connected subtrees on $t$ vertices is bounded by $C n \alpha^t \Delta_*^{t-1}$. Therefore, the total number of such subgraphs is 
\[
C n \alpha^t \Delta_*^{t-1} \binom{\frac{t(t+1)}{2}}{\eta t + 1}
\leq C n \alpha^t \Delta_*^{t-1} \binom{t^2}{2 \eta t}
\leq C n \alpha^t \Delta_*^{t-1} (\frac{et}{2 \eta})^{2 \eta t}.
\]
The probability that such a graph exists in $G$ is at most 
\[
C n \alpha^t \Delta_*^{t-1} (\frac{et}{2 \eta})^{2 \eta t} (\frac{\Delta}{\Delta_*})^{(1 + \eta)t}
\leq \frac{n}{\Delta_*} (C_1 \Delta^2)^t (C_2 \frac{l^2 t^2}{\Delta_*})^{t / 3l}.
\]
Let $A = C_1 \Delta^2$ and $B = C_2\frac{l^2t^2}{\Delta_*}$. The above quantity is at most
\[
\frac{n}{\Delta_*} A^{t}B^{t/3l} = (\frac{n}{\Delta_*} A^{3l} B) (AB^{1/3l})^{t - 3l}.
\]
Assume $t \leq n^{\gamma}$ for some $\gamma \in (0, 0.1)$ and $l = O(\frac{\log n}{\log \Delta})$ be such that 
$\frac{n}{\Delta_*} A^{3l} B = \frac{C_2 l^2t^2 n (C_1\Delta^2)^{3l} }{\Delta_*^2} \leq n^{-0.1}$, which also implies $AB^{1/3l} \leq 1$. Summing over $t = 3l, \dots, n^{\gamma}$, the probability that such a graph exists is bounded by $o(1)$. 

\paragraph{Putting Them Together}
In Section~\ref{subsec:instance}, we mentioned that the resulting graph should be amenable to have a Sherali-Adams solution with a large value, and introduced the notion of path-decomposability to measure it. 
The following lemma of Arora et al.~\cite{ABLT06} shows that our construction satisfies that every subgraph of $G$ induced on at most $t \leq n^{\gamma}$ vertices is $l$-path decomposable.

\begin{lem}
[\cite{ABLT06}]
Let $l \geq 1$ be an integer and $0 < \eta < \frac{1}{3l - 1}$, and let $H$ be a 2-connected graph with $t$ vertices and at most $(1+\eta)t$ edges. Then $H$ contains a path of length at least $l+1$ whose internal vertices have degree 2 in $H$.
\end{lem}

Finally, $\delta$ and $l$ are fixed based on the other parameters to satisfy the requirements of the theorem.

\begin{lem}
There exists $\delta > 0$ and $l$ (depending on $T$, $\epsilon$, $\Delta$, $\mu$, $\gamma$) such that for any $k \leq n^{\delta}$, the following holds.
\label{lem:getdelta}
\end{lem}
\begin{enumerate}
\itemsep=0ex
\item $(1 - \mu)^{\frac{l}{10}} \leq \frac{\mu}{5k}$.
\item Every induced subgraph of $G$ on $(2\Delta)^l k$ vertices is $l$-path decomposable.
\end{enumerate}
\begin{proof}
The first condition is implied by $l \geq C \delta \log n$ for some constant $C$ depending on $\mu$.
The second condition is implied by $(2\Delta)^lk \leq n^{\gamma} \Leftrightarrow 
l \leq C'(\gamma - \delta)\log n$ for another constant $C'$ depending on $\Delta$. 
When we control girth and local sparsity, $l$ is required to be $O(\frac{\log n}{\log \Delta})$. Therefore, by taking $\delta$ a small enough constant depending on $T, \epsilon, \Delta, \mu$, and $\gamma$ (all of which depend on $T, \epsilon$), we can ensure that such $l$ exists.
\end{proof}
Therefore, there exist constants $\Delta, \delta$ and $l = \Theta(\log n)$ (all constants depending on $T, \epsilon, \mu$) that satisfy all the requirements given in the theorem.
\end{proof}

\subsection{Distribution}
\label{subsec:proof_distribution}
\begin{cor}
[Restatement of~\ref{cor:distribution}]
Suppose $G' = (V', E')$ is an $l$-path decomposable graph. Let $L = \lfloor l / 9 \rfloor$; $\mu \in [1 / L, 1]$. Then there exists a random mapping $r : V' \rightarrow [T]$ such that
\begin{enumerate}
\itemsep=0ex
\item If $d := d(u, v) \leq L$ then 
\[
\Pr[r(u) = i, r(v) = i'] = 
\begin{cases}
\frac{(1 - \mu)^d}{(T + 1)} + \frac{1 - (1 - \mu)^d}{(T + 1)^2} & \mbox{if } i \mbox{ and } i' \mbox{ match} \\
\frac{1 - (1 - \mu)^d}{(T + 1)^2} & \mbox{otherwise}
\end{cases}
\]
\item If $d > L$ then $\frac{1 - (1 - \mu)^L}{(T + 1)^2} \leq \Pr[r(u) = i, r(v) = i'] \leq \frac{1 - (1 - \mu)^L}{(T + 1)^2} + \frac{(1 - \mu)^L}{T + 1}$ for any $i, i' \in [T]$.
\end{enumerate}
\end{cor}
\begin{proof}
$r$ is defined by the following process: sample a distribution of multicuts as Theorem~\ref{thm:partition}.
Each piece is a tree, so we can pick an arbitrary vertex $w$ and give a value $l_V(w)$ uniformly from $[T]$ and propagate along the tree to weakly satisfy every edge. Note that the distribution does not depend on the choice of the initial vertex. 

Suppose $d(u, v) \leq L$, which ensures that if $u$ and $v$ are in the same piece, the only path connecting $u$ and $v$ in the piece is the shortest path in $G$. If $i$ and $i'$ are match labels, 
\[
\Pr[r(u) = i, r(v) = i'] = \Pr[u, v \mbox{ in the same piece}]\cdot \frac{1}{T+1} + \Pr[u, v \mbox{ separated}]\cdot \frac{1}{(T+1)^2}.
\]
If $i$ and $i'$ are nonmatching labels, 
\[
\Pr[r(u) = i, r(v) = i'] = \Pr[u, v \mbox{ in the same piece}]\cdot 0 + \Pr[u, v \mbox{ separated}]\cdot \frac{1}{(T+1)^2}.
\]

If $d(u, v) > L$, $\Pr[r(u) = i, r(v) = i']$ is lower bounded by 
$\frac{\Pr[u \mbox{ and } v \mbox{ are separated}]}{(T+1)^2}$, and upper bounded by 
$\frac{\Pr[u \mbox{ and } v \mbox{ are separated}]}{(T+1)^2} + \frac{\Pr[u \mbox{ and } v \mbox{ are not separated}]}{T+1}$. The separation guarantee in Theorem~\ref{thm:partition} proves the lemma.
\end{proof}

\subsection{Embedding}
\label{subsec:proof_embedding}
\begin{lem}
[Restatement of Lemma~\ref{lem:embedding}]
There exist $k(T+1)$ vectors $\left\{ v(i) \right\}_{v \in S, i \in [T]}$ such that $\|v(i)\|_2^2 = \mu + \frac{1}{T + 1}$ and $u(i) \cdot v(i') = \frac{\mu}{2} + \rho(u(i), v(i'))$. 
\end{lem}
\begin{proof}
For each $u(i)$, we construct two vectors $u(i)_1$ and $u(i)_2$ and finally merge them by $u(i) := u(i)_1 \oplus u(i)_2$. $u(i)_2$ is the indicator random variable for the event $l_V(u) = i$, where the distribution follows $\nu_S$. Since $\nu_S$ is based on an actual distribution on the events, the vectors $\left\{ v(i)_2 \right\}_{v \in V, i \in [T]}$ are embeddable into $l_2$ with 
$\|v(i)_2\|_2^2 = \Pr[l_V(v) = i] = \frac{1}{T + 1}$ and $u(i)_2 \cdot v(i')_2 = \nu_S(u(i), v(i'))$. 
The first group of vectors $\left\{ v(i)_1 \right\}_{v \in V, i \in [T]}$ {\em convert} these inner products from $\nu_S$ to $\rho$ with small error.

The following lemma says that a metric space can be isometrically embeddable into $l_2$ if all pairwise distances are similar.
\begin{lem}[\cite{CMM10}] Consider a metric space $(Y, \alpha)$ on $t$ points. If for every two distinct points $u$ and $v$: $|\alpha(u, v) - \beta| \leq \frac{\beta}{2t}$ for some $\beta > 0$, then $(Y, \alpha)$ is isometrically embeddable into $l_2$. 
\label{lem:l2_embedding}
\end{lem}
We add a vector $O$ (so that we have $k(T+1) + 1$ vectors) and set the following distance requirements.

\begin{enumerate}
\item $\|v(i)_1 - O\|_2 = \sqrt{\mu}$ for all $u_i$. 
\item $\|u(i)_1 - v(i')_1\|_2 = \sqrt{\mu - 2 \rho(u(i), v(i')) + 2\nu_S(u(i), v(i'))}$ for all $u(i), v(i')$.
\end{enumerate}

Note that $|\rho(u(i), v(i')) - \nu_S(u(i), v(i'))| \leq \frac{(1 - \mu)^L}{T+1} \leq \frac{\mu}{5(T+1)k}$, where the 
last inequality follows from Theorem~\ref{thm:goodinstance}. 
This implies
\[
| \| u(i)_1 - v(i')_1\|_2 - \sqrt{\mu}| \leq \sqrt{\mu} (1 - \sqrt {1 - \frac{1}{2.5(T+1)k}}) \leq \sqrt{\mu} \cdot \frac{1}{2((T + 1)k+1)}.
\]
By Lemma~\ref{lem:l2_embedding}, there are vectors $\left\{ u(i)_1, v(i)_1 \right\}_i$ and $O$ that meet the above distance requirements. Without loss of generality, assume that $O$ is the origin. Defining $u(i) := u(i)_1 \oplus u(i)_2$ satisfies

\begin{enumerate}
\item $\|u(i)\|_2^2 = \mu + \frac{1}{T + 1}$.
\item $u(i) \cdot v(i') = u(i)_1 \cdot v(i')_1 + 
u(i)_2 \cdot v(i')_2 = \frac{2\mu - \|u(i)_1 - v(i')_1\|_2^2}{2} + \nu_S(u(i), v(i')) = \frac{\mu}{2} + \rho(u(i), v(i'))$.
\end{enumerate}
\end{proof}

\subsection{Rounding}
\label{subsec:proof_rounding}
\begin{lem}
[Restatement of Lemma~\ref{lem:rounding}]
There exists $\mu$ depending on $T$ and $\epsilon$ such that, in the above rounding scheme, the probability that $l_V(u) = 0$ and $l_V(v) = l_A(u, v)$ is at most $\frac{1 - 12\epsilon}{T + 1}$.
\end{lem}
\begin{proof}
For notational simplicity, assume $l_A(u, v) = 0$ --- which is not allowed in actual instances. Then $u(i)$ and $v(i)$ become {\em matching} vectors --- $\rho(u(i), v(i)) = \frac{1 - \mu}{T + 1} + \frac{\mu}{(T + 1)^2}$ and $\rho(u(i), v(j)) = \frac{\mu}{(T + 1)^2}$ for $i \neq j$. The following is the list of all possible inner products between $2(T+1)$ vectors.

\begin{enumerate}
\item $\| u(i)\|_2^2 = \mu + \frac{1}{T + 1}$.
\item $u(i) \cdot u(j) = \frac{\mu}{2}$ for $i \neq j$.
\item $u(i) \cdot v(i) = \frac{\mu}{2} + \frac{1 - \mu}{T + 1} + \frac{\mu}{(T + 1)^2}$.
\item $u(i) \cdot v(j) = \frac{\mu}{2} + \frac{\mu}{(T+1)^2}$ for $i \neq j$.
\end{enumerate}

Even though we used Lemma~\ref{lem:l2_embedding} as a black-box to obtain the current embedding, we can explicitly represent $u(i), v(i)$'s in the Euclidean space. They can be represented as a linear combination of $(T+1) + (T+1)^2 + 2(T+1) + 1$ orthogonal vectors (with different lengths), which can be classified into the following four categories:

\begin{itemize}
\item $a(i)$ for $i \in [T]$: Length $\sqrt{\frac{1 - \mu}{T+1}}$. Denotes the event that $(u, v)$ is not deleted and $l_V(u) = l_V(v) = i$.
\item $b(i, j)$ for $i, j \in [T]$: Length $\sqrt{\frac{\mu}{(T+1)^2}}$. Denotes the event that $(u, v)$ is deleted and $l_V(u) = i$, $l_V(v) = j$.
\item $c(i), c'(i)$ for $i \in [T]$: Length $\sqrt{\frac{\mu}{2}}$. One of them is assigned for each of $2(T+1)$ vectors.
\item $d$: Length $\sqrt{\frac{\mu}{2}}$. Common for all vectors.
\end{itemize}
Let 
\begin{align*}
u(i) := a(i) + \sum_{j} b(i, j) + c(i) + d \\
v(i) := a(i) + \sum_{j} b(j, i) + c'(i) + d.
\end{align*}
It is straightforward to check that the following representation of $u(i)$ and $v(i)$ satisfy all the inner product requirements. 

For each vector $u(i)$, we denote the random variable equal to the inner product of $u(i)$ and $g$ by $U(i)$. Similarly, define $V(i), A(i), B(i, j), C(i), C'(i), D(i)$ for $v(i), a(i), b(i, j), c(i), c'(i), d(i)$ respectively.
Each random variable follows the Gaussian distribution with mean 0 and standard deviation same with the length of the corresponding vector. Furthermore, the inner products of two vectors is the same with the covariance of corresponding random variables. The following lemma shows that our consistent local distributions actually satisfy each edge with probability close to $\frac{1}{T + 1}$, proving Theorem~\ref{thm:gmd}.

\begin{lem}
Fix $i \in [T]$ and $0 < \epsilon < 1/24$. If $\mu \leq \frac{\epsilon^2}{256(T+1) \log^{2} (\frac{T+1}{\epsilon})}$, 
\[
\Pr[l_V(u) =  i, l_V(v) = i] \geq \frac{1 - 12\epsilon}{T+1}.
\]
\end{lem}
\begin{proof}
We compute the probability that $u$ and $v$ are assigned the same label $i$. 

\begin{eqnarray*}
&& Pr[l_V(u) =  i, l_V(v) = i] \\
& \geq & \Pr[A(i) = \max_j [A(j)]] * \Pr[\max_{j}[\sum_{k} B(j, k)], \max_{j}[\sum_{k} B(k, j)], \max_{j}[C(j)], \max_{j}[C'(j)] \\ 
&& \leq \frac{A(i)  - \max_{j \neq i}[A(j)]}{4} | A(i) = \max_j [A(j)]] \\
&\geq& \frac{1}{T+1} \Pr[\max_{j}[\sum_{k} B(j, k)], \max_{j}[\sum_{k} B(k, j)], \max_{j}[C(j)], \max_{j}[C'(j)] \leq \frac{\max_{j}[A(j)]  - \secondmax_{j}[A(j)]}{4}]
\end{eqnarray*}

We argue that the above quantity is close to $\frac{1}{T+1}$ by showing that each of 4 quantities 
\[
\max_{j}[\sum_{k} B(j, k)], \max_{j}[\sum_{k} B(k, j)], \max_{j}[C(j)], \max_{j}[C'(j)]
\]
is greater than $\frac{\max_{j}[A(j)]  - \secondmax_{j}[A(j)]}{4}$ with small probability. 
Note that $\sum_{k} B(j, k)$ follows the Gaussian distribution with mean 0 and variance $\frac{\mu}{T+1}$, which is much less than that of $C(j)$. Since $C(j)$ and $C'(j)$ follow the same distribution, it is enough to show that $\max_{j}[C(j)] > \frac{\max_{j}[A(j)]  - \secondmax_{j}[A(j)]}{4}$ with small probability. The following claim proves the lemma.
\end{proof}

\begin{claim}
Let $0 < \epsilon < 1/4$. If $\mu \leq \frac{\epsilon^2}{256(T+1) \log^{2} (\frac{T+1}{\epsilon})}$,
\[
\Pr[\max_{j}[C(j)] >  \frac{\max_{j}[A(j)]  - \secondmax_{j}[A(j)]}{4}] < 3\epsilon.
\]
\end{claim}
\begin{proof}
The above probability can be rewritten as
\[
\Pr[\sqrt{\frac{\mu}{2}} \max_{j}[g_j] > \sqrt{\frac{1 - \mu}{T+1}} \frac{\max_{j}[g'_j]  - \secondmax_{j}[g'_j]}{4}]
\]
where $g_0, \dots, g_T, g'_0, \dots, g'_T$ are independent standard Gaussian random variables. 

Let $x \geq \sqrt{\mu \log \frac{T+1}{\epsilon}}$. By Lemma~\ref{lem:maxgauss1},
\[
\Pr[\sqrt{\frac{\mu}{2}} \max_{j}[g_j] > x] < \epsilon.
\]
Let $x \leq \frac{\epsilon}{8\sqrt{\log \frac{T+1}{\epsilon}}} \sqrt{\frac{1 - \mu}{T + 1}}$. By Lemma~\ref{lem:maxgauss2}, 
\[
\Pr[\sqrt{\frac{1 - \mu}{T+1}} \frac{\max_{j}[g'_j]  - \secondmax_{j}[g'_j]}{4} < x] < 2\epsilon.
\]
The fact that $\mu \leq \frac{\epsilon^2}{256(T+1) \log^{2} (\frac{T+1}{\epsilon})}$ ensures that there is $x$ that satisfies the both Lemma~\ref{lem:maxgauss1} and~\ref{lem:maxgauss2}. Taking union bound proves the lemma.
\end{proof}
It remains to prove the following two lemmas about Gaussians. 
We prove them in Appendix~\ref{sec:gaussian} using some basic properties of Gaussians.
\begin{lem}
Let $g_1, \dots, g_n$ ($n \geq 2$) be independent standard Gaussian random variables and $0 < \epsilon < 1$.
If $x \geq \sqrt{2 \log \frac{n}{\epsilon}}$, 
\[
\Pr[\max_j [g_j] \leq x] \geq 1 - \epsilon.
\]
\label{lem:maxgauss1}
\end{lem}

\begin{lem}
Let $g_1, \dots, g_n$ ($n \geq 2$) be independent standard Gaussian random variables and $0 < \epsilon < 1/4$.
If $x \leq \frac{\epsilon}{2\sqrt{\log \frac{n}{\epsilon}}}$,
\[
\Pr[\max_j[g_j] - \secondmax_j[g_j] \geq x] \geq (1 - 2\epsilon).
\]
\label{lem:maxgauss2}
\end{lem}
\end{proof}

\section{$(\frac{1}{4} + \Omega(\frac{1}{T}))$-Approximation Algorithm for Generalized Max-Dicut}
\label{sec:maxdicut_algo}
In this section, we propose an approximation algorithm for Generalized Max-Dicut($T$) 
that guarantees $(\frac{1}{4} + \frac{1}{16T})$ fraction of the optimal solution, 
proving Theorem~\ref{thm:gmd_algo}.
It is based on the 2-rounds of the Sherali-Adams hierarchy (also known as the standard LP), defined as the following:
\begin{eqnarray*}
\mbox{maximize} & \sum_{(u, v) \in A} x_{(u, v)} (0, l_A(u, v)) & \\
\mbox{subject to} & \sum_{\alpha \in [T]^S} x_S(\alpha) = 1 & \mbox{ for all } S \subseteq V, |S| \leq 2 \\
& \sum_{j \in [T]} x_{(u, v)}(i, j) = x_u(i) & \mbox{ for all } u \neq v, i \in [T]
\end{eqnarray*}

The algorithm is almost identical to the simple $\frac{1}{4}$-approximation algorithm. For each vertex $v$, independently set $l_V(v) = 0$ with probability $\frac{1 + x_v(0)}{2}$, and $l_V(v) = i$ ($i \neq 0$) with probability $\frac{x_v(i)}{2}$. Equivalently, we assign each vertex 0 with probability half and follow its marginal $x_v$ with probability half. 

For each edge $(u, v) \in A$, let $c = c(u, v) := x_{(u, v)}(0, l_A(u, v))$
so that the value the solution $\left\{ x_S(\alpha) \right\}$ to the LP is $\E_{(u, v)}[c(u, v)] \geq \Opt$. The probability that $(u, v)$ is satisfied is 
\[
(\frac{1 + x_u(0)}{2}) (\frac{x_v(l_A(u, v))}{2}) \geq \frac{c}{4} + \frac{c^2}{4}
\]
since $x_u(0), x_v(l_A(u, v)) \geq c$.
Therefore, the expected fraction of satisfied edges is at least 
\[
\E_{(u, v) \in A} [\frac{c(u, v)}{4} + \frac{c(u, v)^2}{4}] \geq \frac{\Opt}{4} + \frac{\Opt^2}{4} \geq \frac{\Opt}{4} + \frac{\Opt}{16T} 
\]
since $\Opt \geq \frac{1}{4T}$ (focusing on the label with the most edges and finding the maximum dicut with respect to the edges with this label guarantees to satisfy $\frac{1}{4T}$ fraction of edges). 

\section{Details in the Unique Games-Hardness}
\label{sec:proof_ughard}
\subsection{Lemmas about the Dictatorship Test}
\begin{lem}
[Restatement of Lemma~\ref{lem:correlation}]
For any $t$, $\rho(\Omega_1, \Omega_2; \PP^t) \leq \sqrt{\frac{2}{T \delta}}$.
\end{lem}
\begin{proof}
Let $f : \Omega_1 \rightarrow \mathbb{R}$ be the function satisfying $\E[f] = 0$, $\E[f^2] = 1$.
Let $L$ be the Markov operator defined in Section 2.1 of Mossel~\cite{Mossel10} such that
\[
(Lf)(y) = \E[f(X) | Y = y]
\]
for $y \in \Omega_2$ and $(X, Y) \in \Omega_1 \times \Omega_2$ is distributed according to $\PP^t$. 
By Lemma 2.8 of~\cite{Mossel10}, 
\[
\rho(\Omega_1, \Omega_2) = \sup_f \sqrt{\E[(Lf)^2]}.
\]
Let $f(i) = a_i, (Lf)(i) = b_i$ for $i \in [T]$. 
$b_t = a_0$ and all the other $b_i$'s are equal to $\E_{\PP'_1}[f]$, which is equal to 
$(\frac{1}{1 - \frac{1 - \delta}{T}})(\E_{\PP_1}[f] - \frac{1 - \delta}{T} a_0) = 
(\frac{1}{1 - \frac{1 - \delta}{T}})(- \frac{1 - \delta}{T} a_0)$. 

\begin{eqnarray*}
\E[(Lf)^2] &=& (\frac{1 - \delta}{T}) a_0^2 + 
(1 - \frac{1 - \delta}{T}) [(\frac{1}{1 - \frac{1 - \delta}{T}})(- \frac{1 - \delta}{T} a_0)]^2 \\
&=& (\frac{1 - \delta}{T}) a_0^2 + 
(\frac{1}{1 - \frac{1 - \delta}{T}})(\frac{1 - \delta}{T} a_0)^2 \\
&=& (\frac{1 - \delta}{T}) a_0^2 [ 1 +
(\frac{1}{1 - \frac{1 - \delta}{T}})(\frac{1 - \delta}{T})] \\
&\leq& \frac{2}{T} a_0^2 \\
&\leq& \frac{2}{T \delta}
\end{eqnarray*}
Since $\delta a_0^2 \leq \E[f^2] \leq 1$.
\end{proof}

\subsection{Reduction From Unique Games}
\label{sec:ugreduction}
In this subsection, we introduce the reduction from the Unique Games to Generalized Max-Dicut($T$), using the dictatorship test constructed.
We first introduce the Unique Games Conjecture~\cite{Khot02}, which is stated below.
\begin{defi} [Unique Games]
An instance $\mathcal{L} (G(U \cup W, E), [R]^+, \left\{ \pi(u, w) \right\}_{(u, w) \in E})$ of Unique Games consists of a regular bipartite graph $G(U \cup W, E)$ and a set $[R]^+$ of labels. For each edge $(u, w) \in E$ there is a constraint specified by a permutation $\pi(u, w) : [R]^+ \rightarrow [R]^+$. Given a labeling $l : U \cup W \rightarrow [R]^+$, let $\Valug(l)$ be the fraction of labels satisfied by $l$, where an edge $e = (u, w)$ is said to be satisfied if $l(u) = \pi(u, w)(l(w))$. 
Let $\Optug(\mathcal{L}) = \max_l (\Valug(l))$. 
\end{defi}
\begin{conj} [Unique Games Conjecture~\cite{Khot02}] 
\label{conj:ug}
For any constant $\alpha > 0$, there is $R = R(\alpha)$ such that, for a Unique Games instance $\mathcal{L}$ with label set $[R]^+$, it is NP-hard to distinguish between
\begin{itemize}
\item $\Optug(\mathcal{L}) \geq 1 - \alpha$. 
\item $\Optug(\mathcal{L}) \leq \alpha$. 
\end{itemize}
\end{conj}

\begin{thm}
[Restatement of Theorem~\ref{thm:ug-hard}]
Under the Unique Games Conjecture, 
it is NP-hard to approximate Generalized Max-Dicut($T$) on directed acyclic graphs within a factor of $\frac{1}{4} + O(\frac{1}{T^{1/4}})$.
\end{thm}

\begin{proof}

Given an instance of $\mathcal{L} (G(U \cup W, E), [R]^+, \left\{ \pi(v, w) \right\}_{(v, w) \in E})$ of Unique Games, we construct an instance $\calD(\calV, \calA), l_{\calA}$ of Generalized Max-Dicut($T$). For $x \in [T]^R$ and a permutation $\pi : [R]^+ \rightarrow [R]^+$, let $x \circ \pi \in [T]^R$ be defined by $(x \circ \pi)_i = (x)_{\pi^{-1}(i)}$. Let $D = (V, A)$ be the fixed-size graph where the maximum dicut has at most $(\frac{1}{4} + \frac{1}{T^{1/4}})$ fraction of edges. 
\begin{itemize}
\item $\calV = U \times V \times [T]^R$.
\item Sample $w \in W$ uniformly at random and its neighbors $u_1, u_2$ uniformly and independently. 
Sample $t \in [T]^+$, $(v_1, v_2) \in A$, and $x, y \in [T]^R$ from the dictatorship test. Add an edge $((u_1, v_1, x \circ \pi_{u_1, w}), (u_2, v_2, y \circ \pi_{u_2, w}))$ to $\calA$ with label $t$. The weight is equal to the probability that this edge is sampled. 
\end{itemize}

\paragraph{Completeness} Suppose that $\Valug(l) \geq 1 - \alpha$ for some labeling $l : U \cup W \rightarrow [R]^+$.

Set $l_{\calV}(u, v, (x_1, \dots, x_R)) = x_{l(u)}$. 
For $w, u_1, u_2$ sampled as above, with probability $1 - 2 \alpha$, $\pi(u_1, w)^{-1}(l(u_1)) = \pi(u_2, w)^{-1}(l(u_2))$.
In that case, by Lemma~\ref{lem:completeness}, 
\begin{align*}
& \Pr_{v_1, v_2, t, x, y} [l_{\calV}(u_1, v_1, x \circ \pi_{u_1, w}) = 0, l_{\calV}(u_2, v_2, y  \circ \pi_{u_2, w}) = t]  \\
=& \Pr_{v_1, v_2, t, x, y} [(x \circ \pi_{u_1, w})_{l(u_1)}  = 0, (y  \circ \pi_{u_2, w})_{l(u_2)} = t] \\
=& \Pr_{v_1, v_2, t, x, y} [(x)_{\pi(u_1, w)^{-1}(l(u_1))}  = 0, (y)_{\pi(u_2, w)^{-1}(l(u_2))} = t] \\
\geq& \frac{1 - \delta}{T}.
\end{align*}
Therefore, $\Valgmd(l_{\calV}) \geq \frac{(1 - 2\alpha)(1 - \delta)}{T}$. 

\paragraph{Soundness}
For each $u \in U, v \in V$ and $t \in [T]$, let $F_{u, v, t} : [T]^R \rightarrow \left\{ 0, 1 \right\}$ be defined by
\[
F_{u, v, t}(x) = 1 \mbox { if and only if } l_{\calV}(u, v, x) = t.
\]

Similarly, for each $w \in W, v \in V$ and $t \in [T]$, let $H_{w, v, t} : [T]^R \rightarrow [0, 1]$ be the function defined by 
\[
H_{w, v, t}(x) = \E_{(u, w) \in E} [F_{u, v, t}(u, x \circ \pi(u, w))] = 
\Pr_{(u, w) \in E}[l_{\calV}(u, v, x \circ \pi(u, w)) = t].
\]

Suppose that there exists $l_V$ such that $\Valgmd(l_V) \geq \frac{1}{4T} + \frac{5}{T^{5/4}}$. 
For at least $\frac{1}{T^{5/4}}$ fraction of $w$, an edge of $\calA$ sampled by first choosing $w$ is satisfied with probability more than $\frac{1}{4T} + \frac{4}{T^{5/4}}$. 
By Lemma~\ref{lem:soundness}, there exist $\tau$ and $D$, such that, for each such $w$, we have 
$\Inf_i^{\leq d}[H_{w, v, t}] > \tau$ for some $i, v$ and $t$. Set $l_{\calV}(w) = i$. For other $w$'s, choose $l_V(w)$ arbitrarily. 

From the representation of influences in terms of Fourier coefficients (see Khot et al.~\cite{KKMO07}), 
\[
\tau < \Inf_i^{\leq d}[H_{w, v, t}] \leq \E_{(u, w) \in E}[\Inf_{\pi_{u, w}(i)}^{\leq d}[F_{u, v, t}]]
\]
and we conclude that $\tau / 2$ fraction of neighbors $u$ of $w$ have $\Inf_{\pi_{u, w}(i)}^{\leq d}(F_{u, v, t}) \geq \tau / 2$. 
We choose $l_{\calV}(u)$ uniformly from 
\[
\left\{ i : \Inf_{i}^{\leq d}[F_{u, v, t}] \geq \tau / 2 \mbox{ for some } t, v \right\}.
\]
Since $\sum_i \Inf_{i}^{\leq D}[F_{u, t}] \leq d$, 
there are at most $\frac{2(T+1)d|V|}{\tau}$ of candidate $i$'s for each $u$. If $u$ have no candidate, choose $l_V(u)$ arbitrarily. The above strategy satisfies $(\frac{1}{T^{5/4}})(\frac{\tau}{2})(\frac{\tau}{ 2(T+1)D|V|})$ fraction of constraints in expectation. 
Taking $\alpha$ small enough completes the proof of the theorem.
\end{proof}

Now, we present the full proof of our main theorem. 

\begin{thm}
[Restatement of Theorem~\ref{thm:main}]
Under the Unique Games Conjecture, for any $\epsilon > 0$, it is NP-hard to approximate Graph Pricing within a factor of $\frac{1}{4} + \epsilon$. 
\end{thm}
\begin{proof}
Given $\epsilon > 0$, let $T$ large enough so that $\frac{1}{T^{1/4}} < \frac{\epsilon}{2}$. 
Theorem~\ref{thm:ug-hard} tells that it is hard to distinguish
\begin{itemize}
\item Completeness: $\Optgmd \geq \frac{1}{T} - \frac{2}{T^{4/5}} = \frac{1 - O(\epsilon)}{T}$.
\item Soundness: $\Optgmd \leq \frac{1}{4T} + \frac{5}{T^{4/5}} = \frac{1 + O(\epsilon)}{4T}$.
\end{itemize}
Let $t = \frac{T}{\epsilon}$. We can assume that each vertex in the Unique Games instance is of degree at least $t$, 
since duplicating each vertex $v$ into $t$ copies $v_1, ..., v_t$ and duplicating each constraint $(u, v)$ into $t^2$ copies $(u_i, v_j)_{1 \leq i, j \leq t}$ preserves the optimum. Therefore, the instance of Generalized Max-Dicut obtained from the above Unique Games instance will have $\ndeg \geq t$. Theorem~\ref{thm:reduction} shows that it is NP-hard to distinguish 
\begin{itemize}
\item Completeness: $\Optgp \geq \Optgmd = \frac{1 - O(\epsilon)}{T}$.
\item Soundness: $\Optgp \leq \Optgmd + \frac{1}{t} = \frac{1 + O(\epsilon)}{4T} + \frac{\epsilon}{T} = \frac{1 + O(\epsilon)}{4T}$.
\end{itemize}
\end{proof}

\section{Proofs of Lemmas about Gaussians}
\label{sec:gaussian}
Let $\phi(x)$ and $\Phi(x)$ be the probability density function (PDF) and the cumulative distribution function (CDF) of the standard Gaussian, respectively. Let $\tilde{\Phi}(x) = 1 - \Phi(x)$. We begin with the following simple fact about the tail of $\Phi$. 
\begin{lem}[\cite{CMM06}]
For any $t > 0$, $\frac{t}{\sqrt{2\pi}(t^2 + 1)}e^{-\frac{t^2}{2}} < \tilde{\Phi}(t) < \frac{1}{\sqrt{2\pi}t}e^{-\frac{t^2}{2}}$. 
\label{lem:tailgaussian}
\end{lem}

\begin{lem}
[Restatement of Lemma~\ref{lem:maxgauss1}]
Let $g_1, \dots, g_n$ ($n \geq 2$) be independent standard Gaussian random variables and $0 < \epsilon < 1$.
If $x \geq \sqrt{2 \log \frac{n}{\epsilon}}$, 
\[
\Pr[\max_j [g_j] \leq x] \geq 1 - \epsilon.
\]
\end{lem}
\begin{proof}
Note that $x \geq \sqrt{2 \log 2}$, so $\frac{1}{\sqrt{2\pi}{x}} \leq 1$.

\begin{eqnarray*}
&& x \geq \sqrt{2 \log \frac{n}{\epsilon}}  \\
&\Rightarrow& \frac{1}{\sqrt{2\pi}x}\exp(-\frac{x^2}{2}) \leq \frac{\epsilon}{n} \\
&\Rightarrow& 1 - \Phi(x) \leq \frac{\epsilon}{n},
\end{eqnarray*}
where the last inequality follows from Lemma~\ref{lem:tailgaussian}. We can conclude that
\[
\Pr[\max_j [C(j)] \leq x] = \Phi(x)^n \geq (1 - \frac{\epsilon}{n})^n \geq 1 - \epsilon.
\]
\end{proof}

\begin{lem}
[Restatement of Lemma~\ref{lem:maxgauss2}]
Let $g_1, \dots, g_n$ ($n \geq 2$) be independent standard Gaussian random variables and $0 < \epsilon < 1/4$.
If $x \leq \frac{\epsilon}{2\sqrt{\log \frac{n}{\epsilon}}}$,
\[
\Pr[\max_j[g_j] - \secondmax_j[g_j] \geq x] \geq (1 - 2\epsilon).
\]

\end{lem}
\begin{proof}
\begin{eqnarray*}
\Pr[\max_j[g_j] - \secondmax_j[g_j] \geq x]
&\geq& n \int_{-\infty}^{\infty} \Phi[y - x]^{n-1} \phi(y) dy \\
&\geq& n \int_{-\infty}^{b}  \Phi[y - x]^{n-1} \phi(y) dy \quad \mbox{for some } b \mbox{ fixed later} \\
&=& n \int_{-\infty}^{b} \Phi[y - x]^{n-1}  \phi(y - x) \frac{\phi(y)}{\phi(y - x)} dy \\
&\geq& (\inf_{y \in [-\infty, b]} \frac{\phi(y)}{\phi(y - x)}) \int_{-\infty}^{b} n \Phi[y - x]^{n-1} \phi(y - x)  dy \\
&=& (\inf_{y \in [-\infty, b]} \frac{\phi(y)}{\phi(y - x)}) \int_{-\infty}^{b} (\Phi[y - x]^n)' dy \\
&=& (\inf_{y \in [-\infty, b]} \frac{\phi(y)}{\phi(y - x)}) \Phi[b - x]^n
\end{eqnarray*}
Let $b = x + \sqrt{2 \log \frac{n}{\epsilon}}$. By the same argument with Lemma~\ref{lem:maxgauss1}, we have
\begin{eqnarray*}
&& 1 - \Phi[b - x] \leq \frac{\epsilon}{n} \\
&\Rightarrow& \Phi[b - x] \geq 1 - \frac{\epsilon}{n} \\
&\Rightarrow & \Phi[b - x] \geq (1 - \epsilon)^{1/n}
\end{eqnarray*}
Now we bound 
\[
\inf_{y \in [-\infty, b]} \frac{\phi(y)}{\phi(y - x)} = \inf_{y \in [-\infty, b]} \exp(-\frac{y^2}{2} + \frac{(y - x)^2}{2})
= \inf_{y \in [-\infty, b]} \exp(\frac{-2xy + x^2}{2}) = \exp(\frac{-2bx + x^2}{2})
\]
where the last inequality holds since it is monotonically decreasing in $y$. $x \leq \frac{\epsilon}{2\sqrt{\log \frac{n}{\epsilon}}}$ implies 
\begin{eqnarray*}
&& x(x + \sqrt{2 \log \frac{n}{\epsilon}}) \leq \epsilon \\
&\Rightarrow & bx \leq \epsilon \\
&\Rightarrow & \frac{-2bx + x^2}{2} \geq -\epsilon \\
&\Rightarrow & \exp(\frac{-2bx + x^2}{2}) \geq \exp(-\epsilon) \geq 1 - \epsilon
\end{eqnarray*}
Since both $\inf_{y \in [-\infty, b]} \frac{\phi(y)}{\phi(y - x)}$ and $\Phi[b - x]^n$ are at least $1 - \epsilon$, the lemma follows.
\end{proof}

\begin{lem}
[Restatement of Lemma~\ref{lem:gauss_concave}]
Fix $\rho, \alpha \in (0, 1)$. The function $f(x) := \Gamma_{\rho}(\alpha, x)$ is concave. 
\end{lem}
\begin{proof}
Let $Y, Z$ be independent Gaussians and $X := \rho Y + \sqrt{1 - \rho^2} Z$. 
Fix $0 \leq a \leq b$. We will show that $f(a) + f(b) \geq f(a + b)$. 
Let $x = \tilde{\Phi}^{-1}(a + b), y = \tilde{\Phi}^{-1}(b), z = \tilde{\Phi}^{-1}(a), w = \tilde{\Phi}^{-1}(\alpha)$. Note that $x \leq y \leq z$. 
\begin{align*}
& f(a) + f(b) - f(a + b) \\
= &
\Pr[Y \geq y \mbox{ and } X \geq w] +
\Pr[Y \geq z \mbox{ and } X \geq w] -
\Pr[Y \geq x \mbox{ and } X \geq w] \\
= & \Pr[Y \geq z \mbox{ and } X \geq w] -
\Pr[x\leq Y \leq y \mbox{ and } X \geq w] \\
\geq& \Pr[Y \geq z \mbox{ and } Z \geq \frac{w - \rho z}{\sqrt{1-\rho^2}} ]
- \Pr[x\leq Y \leq y \mbox{ and } Z \geq \frac{w - \rho y}{\sqrt{1-\rho^2}} ] \\
=& a (\Pr [Z \geq \frac{w - \rho z}{\sqrt{1-\rho^2}} ] - \Pr [Z \geq \frac{w - \rho y}{\sqrt{1-\rho^2}} ]) \\
\geq& 0
\end{align*}
\end{proof}

\begin{lem}
[Restatement of Lemma~\ref{lem:gauss_bound}]
For large enough $T$ and $\delta = \frac{1}{T^{1/4}}$, the following holds. 
For any $a \in [0, 1], b \in [0, \frac{1}{T}]$ and $\rho \in (0, \sqrt{\frac{2}{T\delta}})$, 
$\Gamma_{\rho}(a, b) \leq ab + \frac{2}{T^{5/4}}$.
\end{lem}
\begin{proof}
Let $Y, Z$ be independent Gaussians and $X := \rho Y + \sqrt{1 - \rho^2} Z$. 
Let $x = \tilde{\Phi}^{-1}(a)$ and $y = \tilde{\Phi}^{-1}(b)$.
By taking $T > 2$, we can assume $b < \frac{1}{2}$ and $y > 0$, while we do not put any assumption on $a$ and $x$. 

\begin{align}
\Gamma_{\rho}(a, b) &= \Pr[X \geq x \mbox{ and } Y \geq y] \notag \\
&\leq \Pr[Z \geq \frac{x - 2\rho y}{\sqrt{1 - \rho^2}} \mbox { and }  y \leq Y \leq 2y] + \Pr[Y \geq 2y] \notag  \\
&\leq \Pr[Z \geq \frac{x - 2\rho y}{\sqrt{1 - \rho^2}} \mbox { and }  Y \geq y] + \Pr[Y \geq 2y] \notag  \\
&\leq b \cdot \tilde{\Phi}(\frac{x - 2\rho y}{\sqrt{1 - \rho^2}}) + \tilde{\Phi}(2y) \label{eq:stability}.
\end{align}
By Lemma~\ref{lem:tailgaussian}, $\tilde{\Phi}(2y)
< \frac{1}{2\sqrt{2\pi}y} \exp(-2y^2) < b^3 < \frac{1}{T^{5/4}}$.

\begin{itemize}
\item $a \geq 1 - \frac{1}{T^{1/4}}$: \eqref{eq:stability} is bounded by $b + \frac{1}{T^{5/4}} \leq (a + \frac{1}{T^{1/4}})b + \frac{1}{T^{5/4}} \leq ab + \frac{2}{T^{5/4}}$. 

\item $b \leq \frac{1}{T^{5/4}}$: \eqref{eq:stability} is bounded by $b + \frac{1}{T^{5/4}} \leq \frac{2}{T^{5/4}}$. 

\item $a \leq 1 - \frac{1}{T^{1/4}}$ and $b \geq \frac{1}{T^{5/4}}$: Note that $x \geq -10 \sqrt{\log T}$ and $y \leq 10 \sqrt{\log T}$. 
Since $\rho \leq \sqrt{\frac{2}{T\delta}} = \frac{\sqrt{2}}{T^{3/8}}$,
\[
(x - 2\rho y) - \sqrt{1 - \rho^2}(x - \frac{1}{T^{1/4}})
\geq 
\begin{cases}
-2\rho y + \frac{1}{2T^{1/4}} \geq 0 & \mbox{if } x \geq 0 \\
\rho^2 x - 2 \rho y + \frac{1}{2T^{1/4}} \geq 0 & \mbox{if } -10 \sqrt{\log T} \leq x \leq 0, \\
\end{cases}
\] 
which shows that $\frac{x - 2\rho y}{\sqrt{1 - \rho^2}} \geq x - \frac{1}{T^{1/4}}$.
Therefore, 
\[
\eqref{eq:stability} \leq b \cdot \tilde{\Phi}(x - \frac{1}{T^{1/4}}) + \frac{1}{T^{5/4}} \leq b(a + \frac{1}{T^{1/4}}) + \frac{1}{T^{5/4}} \leq ab + \frac{2}{T^{5/4}},
\]
where the second inequality follows from $\phi(x) \leq 1$ for all $x \in \RR$. 
\end{itemize}
\end{proof}
\end{document}